\documentclass[onecolumn,pra,nofootinbib,superscriptaddress,notitlepage]{revtex4-1}

\usepackage{amsthm}
\usepackage{amsmath}
\usepackage{appendix}
\usepackage{bbold}
\usepackage{hhline}

\usepackage{tikz}
\usepackage{pgfplots}
\usepackage{pgf}
\usetikzlibrary{patterns}
\usepackage{float}
\usepackage{graphicx}

 \theoremstyle{plain}
  \newtheorem{thm}{Theorem}
  
  \newtheorem{lemma}[thm]{Lemma}
  
  \newtheorem{remark}[thm]{Remark}
  \newtheorem*{thm*}{Theorem}
  \newtheorem*{prop*}{Proposition}
  \newtheorem*{lemma*}{Lemma}
  \newtheorem*{cor*}{Corollary}
  \newtheorem*{remark*}{Remark}
	
	\newtheorem{innercustomthm}{Lemma}
\newenvironment{customlemma}[1]
  {\renewcommand\theinnercustomthm{#1}\innercustomthm}
  {\endinnercustomthm}

\theoremstyle{definition}
\newtheorem{defn}[thm]{Definition}
 \newtheorem*{conj*}{Conjecture}

\usepackage{todonotes}

\begin{document}

\title{Non-signalling parallel repetition using de Finetti reductions}

\author{Rotem Arnon-Friedman}
\affiliation {Institute for Theoretical Physics, ETH-Z\"urich, CH-8093, Z\"urich, Switzerland}
\author{Renato Renner}
\affiliation {Institute for Theoretical Physics, ETH-Z\"urich, CH-8093, Z\"urich, Switzerland}
\author{Thomas Vidick}
\affiliation{Department of Computing and Mathematical Sciences, California Institute of Technology, Pasadena, CA, USA}

\begin{abstract}

In the context of multiplayer games, the parallel repetition problem can be phrased as follows: given a game $G$ with optimal winning probability $1-\alpha$ and its repeated version $G^n$ (in which $n$ games are played together, in parallel), can the players use strategies that are substantially better than ones in which each game is played independently? This question is relevant in physics for the study of correlations and plays an important role in computer science in the context of complexity and cryptography. In this work the case of multiplayer non-signalling games is considered, i.e., the only restriction on the players is that they are not allowed to communicate during the game. For complete-support games (games where all possible combinations of questions have non-zero probability to be asked) with any number of players we prove a threshold theorem stating that the probability that non-signalling players win more than a fraction $1-\alpha+\beta$ of the $n$ games is exponentially small in $n\beta^2$, for every $0\leq \beta \leq \alpha$. For games with incomplete support we derive a similar statement, for a slightly modified form of repetition. 
The result is proved using a new technique, based on a recent de Finetti theorem, which allows us to avoid central technical difficulties that arise in standard proofs of parallel repetition theorems.

\end{abstract}

\maketitle

\section{Introduction}\label{sec:intro}

\subsubsection{Multiplayer games and parallel repetition}

Multiplayer games are relevant in many areas of both theoretical physics and theoretical computer science. In physics, multiplayer games give an intuitive way to study the role and implications of entanglement and correlations, e.g., in the setting of Bell inequalities \cite{bell1964einstein,CHSH}. In computer science such games arise in the context of complexity theory~\cite{FeiGolLovSafSze96JACM,FeiLov92STOC,Dinur07pcp} and cryptography~\cite{BenGolKilWig88STOC,KalaiRR13delegation}. 

In a game $G$, a referee asks each of the cooperating players a question chosen according to a given probability distribution. The players then need to supply answers which fulfil a pre-defined requirement according to which the referee accepts or rejects the answers. In order to do so, they can agree on a strategy beforehand, but once the game begins communication between the players is not allowed. If the referee accepts their answers the players win. The goal of the players is, of course, to maximise their winning probability in the game.  

According to the field of interest, one can analyse any game under different restrictions on the players (in addition to not being allowed to communicate). In classical computer science the players are usually assumed to have only classical resources, or strategies. That is, they can use only local operations and shared randomness. In contrast, one can also consider quantum strategies: before the game starts the players create a multipartite quantum state that can be shared among them. When the game begins each player locally measures their own part of the state and base the answer on their measurement result. It is well-known that sharing quantum entanglement can significantly increase the winning probability in some games \cite{CHSH,aravind2002magic}.  

Another, more general, type of strategies are those where the players can use any type of correlations that do not allow them to communicate, also called non-signalling correlations. That is, the \emph{only} restriction on the players is that they are not allowed to communicate (as will be defined formally later). 

Considering the non-signalling case is interesting for several reasons. A first reason is to minimise the set of assumptions to the mere necessary one. Indeed, if the players are allowed to communicate by sending signals they can win any game. Minimising the set of assumptions can be useful in cryptography when one wishes to get the strongest result possible, i.e., one where the attack strategies of malicious parties are only restricted minimally (as in \cite{hanggi2009quantum,masanes2009universally,masanes2014full} for example). In theoretical physics, non-signalling correlations enable the study of generalised theories possibly beyond quantum theory.  It is also important to mention that, due to their linearity, the non-signalling constraints are often easier to analyse than the quantum or the classical constraints. Therefore, even if additional constraints hold, focusing on the non-signalling ones serves as a way to get first insights into a given problem. 

One of the most interesting questions regarding multiplayer games is the question of parallel repetition. Given a game $G$ with optimal winning probability $1-\alpha$ (using either classical, quantum, or non-signalling strategies), we are interested in analysing the winning probability in the repeated game $G^n$. In $G^n$ the referee gives the players $n$ independent tuples of questions at once, to which the players should reply. The most natural winning criterion is that the players answer a certain fraction $1-\alpha+\beta$ of the $n$ game instances correctly, and one can then ask what is the probability that the players succeed as a function of $\beta$, as the number of repetitions $n$ increases. 

The players can always use the trivial independent and identically distributed (i.i.d.\@) strategy: they just answer each of the $n$ questions independently according to the optimal one-game strategy. In this case the fraction of successful answers is highly concentrated around $1-\alpha$ (alternatively, the probability to win all games simultaneously is $\left(1-\alpha\right)^n$). But can they do significantly better? 

If correlated strategies for $G^n$ are not substantially better than independent ones, even in an asymptotic manner, we learn that ``one cannot fight independence with correlations''. As long as the questions are asked, and the answers are verified, in an independent way, creating correlations between the different answers using a correlated strategy cannot help much. The resulting threshold theorem can then be used, for example, when considering a series of Bell violation experiments performed in parallel, or for hardness amplification in complexity theory and security amplification in classical, quantum and non-signalling cryptography. 

\subsubsection{Related work}

Raz was the first to show in \cite{raz1998parallel} an exponential parallel repetition theorem for classical two-player games. That is, he showed that if the classical optimal winning probability in a game $G$ is smaller than 1, then the probability to win all the games in the repeated game $G^n$, using a classical strategy, decreases exponentially with the number of repetitions $n$. Raz's result was then improved and adapted to the non-signalling case by Holenstein \cite{holenstein2007parallel}. Another improvement was made by Rao in \cite{rao2011parallel}, where a threshold theorem for the classical two-player case was proven: Rao showed that the probability to win more than a fraction $1-\alpha+\beta$ of the games for any $\beta>0$ is exponentially small in the number of repetitions. 

Following the same proof technique as \cite{raz1998parallel,holenstein2007parallel,rao2011parallel},  Buhrman, Fehr and Schaffner recently proved in \cite{buhrman2013parallel} a threshold theorem for multiplayer non-signalling complete-support games, i.e., all possible combinations of questions to the players must have some non-zero probability of being asked. Their threshold theorem was the first result where more than two players were considered. 

The question of parallel repetition in the quantum case is less well understood than its classical and non-signalling versions. All currently known results deal with limited classes of two-player games and no general proof is known. The latest results are given in \cite{jain2013parallel,dinur2013parallel,chailloux2014parallel}, where different assumptions on the probability distribution over the questions of the game are considered. 

\subsubsection{de Finetti theorems in the context of parallel repetition}

The main difficulty in proving a parallel repetition result comes from the, almost arbitrary, correlations between the different questions-answers pairs in the players' strategy for $G^n$: as the players get all the $n$ tuples of questions together they can answer them in a correlated way. In most of the known parallel repetition results (e.g., \cite{raz1998parallel,holenstein2007parallel,rao2011parallel,buhrman2013parallel}) the main idea of the proof is to bound the winning probability for some of the questions, \emph{conditioned} on winning the previous questions. However, as the strategy itself introduces correlations between the different tuples of questions, a large amount of technical work is devoted to dealing with the effect of conditioning on the event of winning the previous questions.

When considering the correlations in a strategy for the repeated game there is one type of symmetry which one can take advantage of, but which is usually virtually ignored -- permutation invariance.  As the game $G^n$ itself is invariant under joint permutation of the tuples of questions and answers, we can restrict our attention to permutation-invariant strategies without loss of generality. Permutation-invariant strategies are strategies which are indifferent to the ordering of the questions given by the referee. That is, the probability of answering a specific set of tuples of questions correctly does not depend on the ordering of the tuples. 

Once we restrict our attention to permutation-invariant strategies, de Finetti theorems seem like a natural tool to leverage for the analysis.
A de Finetti theorem is any type of theorem which relates permutation-invariant states\footnote{Depending on the context, a state can be a probability distribution, a quantum density operator or a conditional probability distribution.} to a more structured state, having the form of a convex combination of i.i.d.\ states, called a de Finetti state. The specific relation between the states depends on the type of theorem. The first de Finetti theorem~\cite{deFinetti69} established that the collection of infinitely exchangeable sequences, in other words those distributions on infinite strings that are invariant under all permutations, exactly coincides with the collection of all convex combinations of i.i.d distributions. Subsequent results establish quantitative bounds on the distance of any permutation-invariant state, or subsystems thereof, from some de Finetti state or an approximation of a de Finetti state~\cite{DiaconisF80finite,raggio1989quantum,caves2002unknown,renner2007symmetry,christandl2007one,brandao2013quantum}. A different form of statement, also called a de Finetti reduction, relates any permutation-invariant state to an explicit de Finetti state by an inequality relation~\cite{christandl2009postselection,arnon2013finetti}. The common feature of all de Finetti theorems is that they enable a substantially simplified analysis of information-processing tasks by exploiting permutation invariance symmetry. Indeed, quantum de Finetti theorems play a significant role in many quantum information problems such as quantum cryptography \cite{christandl2009postselection,leverrier2014composable}, tomography~\cite{christandl2012reliable}, channel capacities~\cite{berta2011reverse} and complexity \cite{brandao2013quantum}.

In the context of games and strategies, de Finetti theorems suggest one may be able to reduce the analysis of general permutation-invariant strategies to the analysis of a de Finetti strategy, i.e., a convex combination of i.i.d.\ strategies. As the behaviour of i.i.d.\ strategies is trivial under parallel repetition, a reduction of this type could simplify the analysis of parallel repetition theorems and threshold theorems. 

Yet, de Finetti theorems were not used in the past in this context, and for a good reason. The many versions of quantum de Finetti theorems (e.g., \cite{renner2007symmetry,christandl2009postselection}) could not have been used as they depend on the dimension of the underlying quantum strategies, while in the quantum multiplayer game setting one does not wish to restrict the dimension. Non-signalling de Finetti theorems, as in~\cite{barrett2009finetti,christandl2009finite}, were also not applicable for non-signalling parallel repetition theorems, as they restrict almost completely the type of allowed correlations in the strategies for the repeated game by assuming very strict non-signalling constraints between the different repetitions, i.e., between the different questions-answers pairs.

In this work we use the recent de Finetti theorem of~\cite{arnon2013finetti}, which imposes no assumptions at all regarding the structure of the strategies (apart from permutation invariance), and is therefore applicable in the context of parallel repetition. This allows us to devise a proof technique which is completely different from the known proofs of parallel repetition results. In particular, at least in the non-signalling case presented here, the conditioning problem described at the beginning of this section disappears completely and the number of players does not play a role in the proof structure.

\subsection{Results and contributions}

The main result presented in this work is a threshold theorem (also called a concentration bound) for the $n$-fold repetition of any $m$-player complete-support game, in which the players are allowed to share any non-signalling strategy. A game is said to have complete-support if all possible combinations of questions to the players have some non-zero probability of being asked. Denote by $w_{\mathrm{ns}}$ the optimal non-signalling winning probability in a game $G$. 
We prove the following theorem.

\begin{thm}\label{thm:final_threshold_theorem}
	For any complete-support game $G$ with $w_{\mathrm{ns}} = 1-\alpha$ there exist $\mathcal{C}_1(G,n)$ and $\mathcal{C}_2(G)$, where $\mathcal{C}_1(G,n)$ is polynomial in the number of repetitions $n$, such that for every $0<\beta\leq\alpha$ and large enough $n$, the probability that non-signalling players win more than a fraction $1-\alpha+\beta$ of the $n$ questions in the repeated game $G^{n}$ is at most $\mathcal{C}_1(G,n)\exp \left[ - \mathcal{C}_2(G) n \beta^2 \right]$.
\end{thm} 

That is, for sufficiently many repetitions the probability to win more than a fraction $1-\alpha+\beta$ of the $n$ games is exponentially small. 
The constant $\mathcal{C}_1(G,n)$ is such that $\mathcal{C}_1(G,n) < 10 m |\mathcal{Q}||\mathcal{A}| \left(n+1\right)^{2(|\mathcal{Q}||\mathcal{A}|-1 )}$ where $m$ is the number of players, and $|\mathcal{Q}|$ and $|\mathcal{A}|$ are the number of possible questions and answers, respectively, in $G$. $\mathcal{C}_2(G)$ is a finite constant that can be computed by solving the polynomial-size linear program given in Equation \eqref{eq:dual}. 
A sufficient condition on the number of repetitions for the bound in the theorem to hold is $n = \Omega\big(|\mathcal{Q}||\mathcal{A}|\frac{\mathcal{C}_2}{\beta^2} \ln^2(|\mathcal{Q}||\mathcal{A}|\frac{\mathcal{C}_2}{\beta})\big)$. We refer to Equation \eqref{eq:number_of_repet_bound}, and the choice of constants made around Equation \eqref{eq:choice_of_constants}, for a more precise bound.

There are two main differences between the exponential bound given in the threshold theorem of \cite{buhrman2013parallel} (Theorem~15 therein) and the bound we give here. First, while our bound suffers from the polynomial dependency on the number of repetitions in $\mathcal{C}_1(G,n)$ (which is inherent to the use of a de Finetti reduction), there is no such dependency in \cite{buhrman2013parallel}. As the number of repetitions goes to infinity, however, the exponential factor quickly dominates. Both our constant $\mathcal{C}_2(G)$ and the constant $\mu$ in Theorem~15 from~\cite{buhrman2013parallel} depend on the size of the game through a certain linear program (see the proof of Lemma~\ref{lem:distance_transformation} and the discussion that follows it in this paper, and the proof of Proposition~18 in~\cite{buhrman2013parallel}), making a direct comparison difficult. Another point of comparison between the bounds is the dependency on $\beta$: we obtain the optimal (as follows from optimal formulations of the Chernoff bound) dependency $\beta^2$, as compared to $\beta^4$ in~\cite{buhrman2013parallel}. As far as we are aware, this is the first threshold theorem where optimal dependency on $\beta$ is achieved (see also~\cite{rao2011parallel}).  

Theorem \ref{thm:final_threshold_theorem} applies to complete-support games. The result is extended in two different directions. First, based on ideas from \cite{ito2010polynomial}, we show in Appendix \ref{app:two_player_extension} that when considering two-player games \emph{without} complete-support Theorem \ref{thm:final_threshold_theorem} still holds. Second, for general multiplayer games we consider in Appendix \ref{app:general_extension} a small modification of the repetition procedure. Instead of the usual parallel repetition procedure, in which $n$ tuples of questions are chosen according to the game distribution $Q$, we change the distribution of questions in the repeated game by sometimes (with small positive probability $\eta$) asking the players a tuple of questions $q$ which does not appear in the original game $G$. We call such questions ``dummy questions''; for these questions any answer from the players is accepted. The remaining questions, for which $Q(q)>0$, are called the ``real questions'' and the modified game is denoted by $\tilde{G}^{n}$. We prove the following threshold theorem:
\begin{thm}\label{thm:mod_threshold_theorem}
	For any game $G$ with $w_{\mathrm{ns}} = 1-\alpha$ there exist $\mathcal{C}_1(G,n)$ and $\mathcal{C}_2(G)$, where $\mathcal{C}_1(G,n)$ is polynomial in the number of repetitions $n$, such that for every $0<\beta\leq\alpha$ and large enough $n$, the probability that non-signalling players win more than a fraction $1-\alpha+\beta$ of the \emph{real questions} in the modified repeated game $\tilde{G}^{n}$ is at most $\mathcal{C}_1(G,n)\exp \left[ - \mathcal{C}_2(G) n \beta^2 \right]$.
\end{thm} 

The constants $\mathcal{C}_1(G,n)$ and $\mathcal{C}_2(G)$ have the same form as in Theorem \ref{thm:final_threshold_theorem}, but they now depend also on the perturbation $\eta$ of the original questions distribution. For more details on the definition of $\tilde{G}^n$ and the proof of Theorem~\ref{thm:mod_threshold_theorem} see Appendix~\ref{app:general_extension}. 

A similar modification was previously considered in both classical \cite{feige1994two} and quantum \cite{kempe2011parallel} parallel repetition theorems, where the repetitions in which dummy questions are selected were called ``confusion rounds''. For many applications this modification is harmless, especially as the success probability of ``honest'' players is not affected by it. However, it is important to note that Theorem~\ref{thm:mod_threshold_theorem} only holds for the modified form of repetition of the original game.

In addition to the bounds themselves our, perhaps most important, contribution in this work is the, arguably simpler, proof technique. While most of the known parallel repetition results build on the proof technique of \cite{raz1998parallel} we give a completely different proof, with ideas based on de Finetti theorems and tomography (as explained in the next section). Our proof technique allows us to avoid the usual difficulties which arise in proofs of parallel repetition theorems, such as conditioning on some of the questions and answers or considering an arbitrary number of players. In this sense our proof can be seen as more natural than previous proofs, and therefore more likely to be extendable to the classical and quantum multiplayer cases as well.

\subsection{Proof idea and techniques}\label{sec:proof_idea}

The goal of this section is to give the reader an intuitive understanding of the proof idea and techniques. The formal and more technical implementations of these ideas are given in the following sections. Nevertheless, the following two definitions are needed. 

\begin{defn} [Multiplayer game]\label{def:game}
	An $m$-player game  $G=(\mathcal{Q},\mathcal{A},Q,R)$ is defined by a set of possible tuples of questions $\mathcal{Q}$ together with a probability distribution $Q:\mathcal{Q}\rightarrow [0,1]$ (according to which the referee choses the questions) over it, a set of possible tuples of answers $\mathcal{A}$ and a winning condition $R:\mathcal{Q}\times\mathcal{A}\rightarrow\{0,1\}$. An $m$-tuple of questions $q = (q^1,q^2,\dotsc,q^m) \in \mathcal{Q}$ describes the questions given to the different players. Similarly an $m$-tuple of answers $a =  (a^1,a^2,\dotsc,a^m) \in \mathcal{A}$ describes the answers given by the different players.
\end{defn}

\begin{defn} [Strategy]\label{def:strategy}
	A strategy for an $m$-player game  $G=(\mathcal{Q},\mathcal{A},Q,R)$ is a conditional probability distribution $\mathrm{O}_{A|Q}:\mathcal{A}\times\mathcal{Q}\rightarrow[0,1]$, i.e., $\sum_a \mathrm{O}_{A|Q}(a|q)=1$ for all $q\in \mathcal{Q}$. Similarly, a strategy for a repeated game $G^n$ is a conditional probability distribution denoted by $\mathrm{P}_{\vec{A}|\vec{Q}}:\mathcal{A}^n\times\mathcal{Q}^n\rightarrow[0,1]$. 
\end{defn}

Throughout the proof strategies for the game $G$ are denoted by $\mathrm{O}_{A|Q}$ and strategies for the repeated game $G^n$ are denoted by $\mathrm{P}_{\vec{A}|\vec{Q}}$.

\subsubsection{Permutation invariance and de Finetti theorems}

The first trivial, but crucial, observation made is that when considering strategies for the repeated game, one can concentrate without loss of generality on permutation-invariant strategies. Permutation-invariant strategies are indifferent to the ordering of the tuples of questions given by the referee. That is, the referee can ask the players to answer  $q_1,q_2,q_3$ or $q_2,q_3,q_1$ (each $q_i$ is an $m$-tuple); in both cases the winning probability will be the same if the players are using a permutation-invariant strategy. Note that the permutation changes only the order of the tuples of questions. In particular, the players themselves are not being permuted and the questions of all players are permuted in exactly the same way (see Definition \ref{def:permutation} and Lemma \ref{lem:prem_wlog} for the formal argument).  

Considering only permutation-invariant strategies allows us to use the de Finetti theorem of \cite{arnon2013finetti} which relates any permutation-invariant strategy to a de Finetti strategy.  
The exact statement of the de Finetti theorem will only be relevant later. For now, using just the intuition of de Finetti theorems, one can think of any permutation-invariant strategy as being a convex combination of i.i.d.\ strategies. That is\footnote{We emphasise once again that this is not a quantitive statement that we claim to be correct. This is just useful as an intuitive way of understanding the proof idea.}, 
\begin{equation}\label{eq:intuition}
	\mathrm{P}_{\vec{A}|\vec{Q}} \approx \int \mathrm{O}_{A|Q}^{\otimes n} \mathrm{d}\mathrm{O}_{A|Q}
\end{equation}
where $\mathrm{d}\mathrm{O}_{A|Q}$ is some measure on the space of one-game strategies and $\mathrm{O}_{A|Q}^{\otimes n}$ is a product of $n$ identical strategies~$\mathrm{O}_{A|Q}$.

Unfortunately, the convex combination itself (meaning, the measure $\mathrm{d}\mathrm{O}_{A|Q}$) is unknown.  Moreover, even though we assume that the strategy $\mathrm{P}_{\vec{A}|\vec{Q}}$ does not allow the $m$ players to communicate, i.e., it is non-signalling, the convex combination might still include signalling parts, i.e., signalling $\mathrm{O}_{A|Q}$. Indeed, in general, a convex combination of signalling strategies can still be non-signalling. 

For the non-signalling parts of the convex combination one can easily prove a strong parallel repetition or threshold theorem. These parts are just i.i.d.\ non-signalling strategies. The only thing which is left to prove is therefore that the \emph{signalling} part of the convex combination of Equation \eqref{eq:intuition} has an exponentially small weight\footnote{As mentioned above, this statement does not hold for an arbitrary decomposition of a non-signalling strategy as a convex combination of other strategies. We will crucially use the fact that each term in the convex combination has an i.i.d.\ structure.}. We find this question interesting by itself, and of course, the same question can be asked in the classical and quantum case -- given a classical or quantum strategy $\mathrm{P}_{\vec{A}|\vec{Q}}$, what is the weight of the non-classical or non-quantum i.i.d.\ parts in the convex combination?

\subsubsection{Bounding the signalling part}

As the convex combination itself in Equation \eqref{eq:intuition} is unknown, one cannot just calculate the weight of the signalling part. We therefore take a more operational approach, following ideas from quantum tomography \cite{christandl2012reliable}. 

Consider a particular (unknown) part $\mathrm{O}_{A|Q}^{\otimes n}$ of the convex combination and divide the $n$ copies of the strategy $\mathrm{O}_{A|Q}$ into two groups -- a test group consisting of $n/2$ out of the $n$ copies, and a game group of $n/2$ copies. The general idea is to use the test copies to get an estimation $\mathrm{O}^{\mathrm{EST}}_{A|Q}$ of the strategy $\mathrm{O}_{A|Q}$, which will then help us in proving our claims.

More specifically, we are interested in knowing whether $\mathrm{O}_{A|Q}$ is signalling or not (if it is non-signalling then its winning probability in $G$ is obviously bounded by the optimal non-signalling winning probability $1-\alpha$). For this we define a signalling measure and an operational (and hypothetical) signalling test. Given questions and answers which are distributed according to the $n/2$ copies of $\mathrm{O}_{A|Q}$ and $Q$, the signalling test will create an estimation $\mathrm{O}^{\mathrm{EST}}_{A|Q}$ and calculate its signalling value. If the signalling value is above a certain threshold the test will accept, and otherwise it will reject. 

In order to bound the weight of the signalling part in Equation \eqref{eq:intuition} one can bound the probability that the signalling test accepts. To prove that the acceptance probability is exponentially small we use a combination of two lemmas, which we call the weak and the strong lemma. These lemmas are based on a special guessing game that we construct and on applications of the de Finetti theorem. Both lemmas together show that if the probability of the test accepting is too high, then the original strategy $\mathrm{P}_{\vec{A}|\vec{Q}}$ must have been signalling -- a contradiction.

\subsubsection{From intuition to practice}

In practice, the de Finetti theorem \cite{arnon2013finetti} is an inequality relation between any permutation-invariant strategy and a given de Finetti strategy (see Lemma \ref{lem:deFinetti_reduction}) which does not imply Equation \eqref{eq:intuition}. As a consequence, considering the test copies and game copies as above does not directly make sense. Nevertheless, we can follow similar ideas by considering the questions-answers pairs in a specific instance of the repeated game. That is, every time the game is played using a strategy $\mathrm{P}_{\vec{A}|\vec{Q}}$, we divide the data, the questions and answers, of this specific run into two groups -- test data and game data, consisting of $n/2$ tuples of questions-answers pairs each. Our goal is then to bound the winning frequency in the game data, while the test data is relevant for the hypothetical signalling tests (see also Figure~\ref{fig:division_game_test} in the following section). 

The rest of the paper is organised as follows. We start with some preliminaries in Section \ref{sec:prelim}. In Section \ref{sec:signalling} we first consider and explain the concept of non-signalling strategies, then define our signalling measures and signalling tests in a formal way and present their important properties. Section \ref{sec:deFinetti} is devoted to de Finetti and permutation-invariant strategies. Finally, in Section \ref{sec:threshold} we connect all the relevant tools together using the weak and the strong lemmas, and prove our main theorem, Theorem \ref{thm:final_threshold_theorem} (the extension of the theorem to games with incomplete-support is relatively straightforward and is given in Appendix \ref{app:extension}). We conclude in Section \ref{sec:conclusions} with open questions and a discussion of possible extensions to the classical and quantum case. 

\section{Preliminaries}\label{sec:prelim}

Throughout the proof many constants and parameters are used. For convenience, apart from introducing them when necessary, we list all of them together with their role in Table \ref{tb:parameters_table} in the end of Section~\ref{sec:threshold}.  

We use the letters $q,r$ and $s$ to denote tuples of questions and $a$ and $b$ to denote tuples of answers. In the following we define the notation using only $q$ and $a$.

\subsection{Games and strategies}

In this work we consider a general $m$-player game  $G=(\mathcal{Q},\mathcal{A},Q,R)$ as defined in Definition \ref{def:game} in the previous section. A strategy for a game $G$ is described by a conditional probability distribution $\mathrm{O}_{A|Q}:\mathcal{A}\times\mathcal{Q}\rightarrow[0,1]$ as defined in Definition \ref{def:strategy}. For the joint questions-answers distribution we use $\mathrm{O}_{AQ}= Q\times \mathrm{O}_{A|Q}$. 

\begin{defn}[Winning probability]\label{def:win_prob}
	The winning probability of a strategy $\mathrm{O}_{A|Q}$ in game $G=(\mathcal{Q},\mathcal{A},Q,R)$ is given by $w\left( \mathrm{O}_{A|Q} \right) =\sum_{q,a} Q(q) R(q,a) \mathrm{O}_{A|Q}(a|q)$. 
\end{defn}

We use the following definition to measure the distance between two one-game strategies. 
\begin{defn}[Distance measure]\label{def:trace_distance}
	The distance between $\mathrm{K}_{A|Q}$ and $\mathrm{R}_{A|Q}$ is defined as
	\[
		\big|\mathrm{K}_{A|Q}-\mathrm{R}_{A|Q}\big|_1 =  \mathbb{E}_{q\in \mathcal{Q}} \sum_{a\in A} \big| \mathrm{K}_{A|Q}(a|q) - \mathrm{R}_{A|Q}(a|q) \big| 
	\]
	where the $m$-tuples of questions $q\in \mathcal{Q}$ are distributed according to $Q$ defined by the game $G$.
\end{defn}
Note that this is not the standard definition -- instead of a maximisation over the tuples of questions as in the usual definition of the trace distance we consider the average over the tuples according to the game distribution. Therefore, the distance between the strategies depends on the specific game $G$ considered. 

In the repeated game $G^n$ the referee asks each of the players $n$ questions, all at once. The questions are chosen according to the distribution $Q^{\otimes n}$, i.e., independently using $Q$. The answers are then checked independently according to the winning condition $R$. A strategy for the repeated game $G^{n}$ is denoted by $\mathrm{P}_{\vec{A}|\vec{Q}}:\mathcal{A}^{n}\times\mathcal{Q}^{n}\rightarrow[0,1]$ and the joint questions-answers distribution is then $P_{\vec{A}\vec{Q}}=Q^{\otimes n}\times\mathrm{P}_{\vec{A}|\vec{Q}}$. When the distributions are clear from the context we sometimes omit the subscripts and write just $\mathrm{O}$ and $\mathrm{P}$.

When considering many questions-answers pairs in the repeated game we denote all the questions and answers as vectors $\vec{q},\vec{a}$. We use a subscript index as in $\vec{q}_j$ to denote the $j$'th tuple of questions given to the players. We denote by $\mathrm{O}_{A|Q}^{\otimes n}$ a product of $n$ identical strategies~$\mathrm{O}_{A|Q}$. That is, $\mathrm{O}_{A|Q}^{\otimes n}$ is defined according to $\mathrm{O}_{A|Q}^{\otimes n}(\vec{a}|\vec{q}) = \prod_{j=1}^{n} \mathrm{O}_{A|Q}(a_j|q_j)$ for all $\vec{a},\vec{q}$. 

For any $m$-tuple of questions $q = (q^1,\dotsc,q^m) \in \mathcal{Q}$ and any $i\in[m]=\{1,\cdots,m\}$ we denote by $q^i$, using a superscript index, the question given to the $i$'th player by the referee, and by $q^{\bar{i}}=(q^1,\dotsc q^{i-1},q^{i+1},\cdots,q^m)$ the $(m-1)$-tuple of questions given to all the players but~$i$. Similarly, for a subset $I\subset [m]$, $q^I$ denotes the questions given to all the players $i\in I$ and $q^{\bar{I}}$ denotes the complementary set of questions, i.e., the questions given to all the players $i\notin I$. An analogous notation is used for the answers.  Similarly, when considering many questions-answers paris, $\vec{q}^i$ denotes \emph{all} the questions given to the $i$'th player, and so on.  

A tuple of questions $q=(q^1,\dotsc q^{i-1},q^i,q^{i+1},\cdots,q^m)$ can then be also written as $(q^i,q^{\bar{i}})$ where it is understood which player gets which question. Therefore in this notation $Q(q^i,q^{\bar{i}})=Q(q)$ and similarly $\mathrm{O}(a^i,a^{\bar{i}}|q^i,q^{\bar{i}})= \mathrm{O}(a|q)$. Moreover,  $Q(q^i|q^{\bar{i}})= \frac{Q(q^i,q^{\bar{i}})}{\sum_{r^{i}}Q(r^{i},q^{\bar{i}})}$ denotes the probability that the $i$'th player receives question $q^i$ \emph{given} that the other players receive $q^{\bar{i}}$.   

In the following we prove Theorem \ref{thm:final_threshold_theorem}, which applies to games with complete-support. A game has complete-support if every possible combination of questions to the players has some non-zero probability according to the question distribution $Q$. Formally, 
\begin{defn}[Complete-support game]
	An $m$-player game has complete-support if for every possible combination of questions to the players $q^1,\dotsc,q^m$ (i.e., $q^1,\dotsc,q^m$ such that for all $i\in[m]$ there exist $s^{\bar{i}}$ for which $Q\left((s^1,\dotsc,s^{i-1},q^i,s^{i+1},\dotsc,s^m)\right) > 0$), $Q(q)> 0$. 
\end{defn} 

\subsection{Estimated strategies}\label{sec:estimated_strategies}

The specific questions and answers in one run of the repeated game $\vec{q},\vec{a}$ are also called the data of the game. As mentioned in the previous section, the data $\vec{q},\vec{a}$ is divided into two disjoint sets which we call the test data and the game data. We denote the $n/2$ tuples of test questions and answers by $\vec{q^t}, \vec{a^t}$ respectively and the $n/2$ tuples of game questions and answers by $\vec{q^g}, \vec{a^g}$ respectively. Using this notation $\vec{q}$ is the concatenation of $\vec{q^t}$ and $\vec{q^g}$ and $\vec{a}$ is the concatenation of $\vec{a^t}$ and $\vec{a^g}$. Note that although we denote here the test questions as appearing before the game questions, they are indistinguishable from one another, as they are chosen according to the exact same distribution $Q$. Had this not been the case, the permutation invariance symmetry would have been broken.   

Given the test data $\vec{q^t}, \vec{a^t}$ we create an estimation $\mathrm{O}^{\mathrm{EST1}}_{A|Q}$ of a one-game strategy in the following way. For every tuple of questions $q\in \mathcal{Q}$ and answers $a\in \mathcal{A}$ let $f^q_a$ be the frequency of the tuple of answers $a$ in $\vec{a^t}$ restricted to the indices $j\in[n/2]$ for which $\vec{q^t}_j = q$ (if $q$ did not appear at all set $f^q_a=0$). Then define $\mathrm{O}^{\mathrm{EST1}}_{A|Q}$ such that $\mathrm{O}^{\mathrm{EST1}}_{A|Q}(a|q)=f^q_a$. 

Similarly $\mathrm{O}^{\mathrm{EST2}}_{A|Q}$ is created in the same way, using the game data $\vec{q^g}, \vec{a^g}$ (see Figure~\ref{fig:division_game_test}). Note that since $\mathrm{P}_{\vec{A}|\vec{Q}}$ might be signalling between the different $n$ tuples of questions, both $\mathrm{O}^{\mathrm{EST1}}_{A|Q}$ and $\mathrm{O}^{\mathrm{EST2}}_{A|Q}$ can depend also on the other questions which are not considered in the estimation process. 

\begin{figure}
\begin{centering}
	\begin{tikzpicture}
	
		\draw[decorate,decoration={brace,amplitude=8pt}] (-0.55,1.25) -- (-0.5,2.55) node [midway,xshift=-1.5cm,align=center] {distributed \\ according to \\ $\mathrm{P}_{\vec{A}\vec{Q}}$};
		
		\draw (0,1.5) node {$\vec{a}:$};
		\node (rect) at (2,1.5) [draw,minimum width=3cm, minimum height=0.6cm] {$\vec{a^t}$};
		\node (rect) at (5.25,1.5) [draw,minimum width=3cm, minimum height=0.6cm] {$\vec{a^g}$};
		
		\draw (0,2.25) node {$\vec{q}:$};
		\node (rect) at (2,2.25) [draw,minimum width=3cm, minimum height=0.6cm] {$\vec{q^t}$};
		\node (rect) at (5.25,2.25) [draw,minimum width=3cm, minimum height=0.6cm] {$\vec{q^g}$};
		
		\draw[decorate,decoration={brace,mirror,amplitude=10pt}] (0.5,0.85) -- (3.5,0.85) node [midway,yshift=-1cm,align=center] {defines $\mathrm{O}^{\mathrm{EST1}}_{A|Q}$};
		
		\draw[decorate,decoration={brace,mirror,amplitude=10pt}] (3.75,0.85) -- (6.75,0.85) node [midway,yshift=-1cm,align=center] {defines $\mathrm{O}^{\mathrm{EST2}}_{A|Q}$};
	
	\end{tikzpicture}
\par\end{centering}
\caption{Division to test and game data} \label{fig:division_game_test}
\end{figure}

To evaluate the accuracy of the estimation process described above we will use the following lemma~-- an application of Sanov's theorem (see, e.g., \cite{cover2012elements} Section 11.4) to our scenario. 
\begin{lemma}\label{lem:chernoff}
	Let $\delta(l,\epsilon)=(l+1)^{|\mathcal{A}|\cdot|\mathcal{Q}|-1}e^{-l\epsilon^2/2}$. Then for every i.i.d. strategy $\mathrm{O}_{A|Q}^{\otimes l}$,
	\[
		\mathrm{Pr}_{\vec{a},\vec{q}\sim\mathrm{O}^{\otimes l}_{AQ}}\left[ |\mathrm{O}^{\mathrm{EST}}_{A|Q} -  \mathrm{O}_{A|Q}|_1 > \epsilon \right] \leq \delta(l,\epsilon) 
	\]
	where $\mathrm{O}^{\mathrm{EST}}_{A|Q}$ is estimated as above from the data $\vec{a},\vec{q}$. 
\end{lemma}

\subsection{Linear programs}

Linear programs (see, e.g., \cite{schrijver1998theory}) are a useful tool when considering non-signalling games, as the non-signalling constraints are linear. The following results regarding the sensitivity of linear programs will be of use for us. 

\begin{lemma}[Sensitivity analysis of linear programs, \cite{schrijver1998theory} Section 10.4]\label{lem:sensitivity1}
	Let $\max\{c^Tx | Ax \leq b\}$ be a primal linear program and $min\{b^Ty|A^Ty=c,y\geq 0\}$ its dual. Denote the optimal value of the programs by $w$ and the optimal dual solution by $y^{\star}$. Then the optimal value of the perturbed program $w_e = max\{c^Tx | Ax \leq b+e\}$ for some perturbation $e$ is bounded by $w_e \leq w + e^Ty^{\star}$. 
\end{lemma}

\begin{lemma}[Dual optimal solution bound, \cite{schrijver1998theory} Section 10.4]\label{lem:sensitivity2}
	Let $A$ be an $r_1\times r_2$-matrix and let $\Delta$ be such that for each non-singular submatrix $B$ of $A$ all entries of $B^{-1}$ are at most $\Delta$ in absolute value. Let $c$ be a row vector of dimension $r_2$ and let $y^{\star}$ be the optimal dual solution of $min\{b^Ty|A^Ty=c,y\geq 0\}$. Then 
	\[
		 \kappa = \sum_{j=1}^{r_1} |y^{\star}_j| \leq  r_2\Delta \sum_{j=1}^{r_2} |c_j| \;.
	\]
\end{lemma}

\section{Detecting signalling}\label{sec:signalling}

\subsection{The non-signalling constraints}

We start by defining a non-signalling strategy. To simplify notation we define it using one-game strategies $\mathrm{O}_{A|Q}$. The definition is identical for the strategies $\mathrm{P}_{\vec{A}|\vec{Q}}$.

\begin{defn}[Non-signalling strategy]
	An $m$-player strategy $\mathrm{O}_{A|Q}$ is called non-signalling if for any set of players $I\subset [m]$, 
	\[
		\forall a^{\bar{I}},q^{\bar{I}},q^I,r^{I} \quad \mathrm{O}_{A|Q}(\circ,a^{\bar{I}}|q^I,q^{\bar{I}}) =  \mathrm{O}_{A|Q}(\circ,a^{\bar{I}}|r^{I},q^{\bar{I}}) 
	\]
	where $\circ$ denotes a marginal, e.g., $\mathrm{O}_{A|Q}(\circ,a^{\bar{I}}|q^I,q^{\bar{I}}) = \sum_{a_i|i\in I} \mathrm{O}_{A|Q}(a|q^I,q^{\bar{I}})$. 
\end{defn}

Alternatively, one can define a non-signalling strategy using a set of linearly independent non-signalling constraints from which all the constraints in the above definition follow.  
\begin{lemma}[Lemma 2.7 in \cite{hanggi2010device}]
	An $m$-player strategy $\mathrm{O}_{A|Q}$ is non-signalling if and only if for any player $i \in [m]$, 
	\begin{equation}\label{eq:normal_ns_def}
		\forall a^{\bar{i}},q^{\bar{i}},q^i,r^{i} \quad  \mathrm{O}_{A|Q}(\circ,a^{\bar{i}}|q^i,q^{\bar{i}}) =  \mathrm{O}_{A|Q}(\circ,a^{\bar{i}}|r^{i},q^{\bar{i}}) \;.
	\end{equation}
\end{lemma}

From Equation \eqref{eq:normal_ns_def} it is clear that for every $i$ and $q^{\bar{i}}$ the marginal states $\mathrm{O}_{A|Q}(\circ,a^{\bar{i}}|q^i,q^{\bar{i}})$ are all equivalent and independent of $q^i$. Therefore another equivalent formulation of the non-signalling constraints is given by 
\[
	\forall a^{\bar{i}},q^{\bar{i}},q^i \quad  \mathrm{O}_{A|Q}(\circ,a^{\bar{i}}|q^i,q^{\bar{i}}) =  \sum_{r^{i}}Q(r^{i}|q^{\bar{i}})\mathrm{O}_{A|Q}(\circ,a^{\bar{i}}|r^{i},q^{\bar{i}}) \;. 
\]

Here we defined the marginal, which is independent of $r^{i}$, as an average over the different $\mathrm{O}_{A|Q}(\circ,a^{\bar{i}}|r^{i},q^{\bar{i}})$, where the average is taken according to the distribution of the game question $Q$. It is easy to verify that this condition is equivalent to Equation \eqref{eq:normal_ns_def}. 

We can now write the optimisation problem of finding the optimal winning probability in a complete-support game $G$ using a non-signalling strategy as the following linear program over the variables $\mathrm{O}(a|q)$:
\begin{subequations} \label{eq:linear_non_relaxed}
\begin{align}
	\max \quad &\sum_{q,a} Q(q) R(q,a) \mathrm{O}(a|q) \label{eq:linear1_obj} \\
	\text{s.t.} \quad&   Q(q^i , q^{\bar{i}})\left[ \mathrm{O}(\circ,a^{\bar{i}}|q^i , q^{\bar{i}}) - \sum_{r^{i}} Q(r^{i}|q^{\bar{i}})\mathrm{O}(\circ,a^{\bar{i}} | r^{i},q^{\bar{i}})\right] = 0  &\forall i, q^i,q^{\bar{i}},a^{\bar{i}} \label{eq:linear1_ns} \\
	&\sum_{a} \mathrm{O}(a|q) = 1  &\forall q \label{eq:linear1_nornmal} \\
	& \mathrm{O}(a|q) \geq 0 &\forall a,q \label{eq:linear1_positive}
\end{align} 
\end{subequations}
The objective function, Equation \eqref{eq:linear1_obj}, is exactly the winning probability of the game using strategy $\mathrm{O}(a|q)$ as defined in Definition \ref{def:win_prob}. Equations \eqref{eq:linear1_nornmal} and \eqref{eq:linear1_positive} are the normalisation and positivity constraints on the strategy~$\mathrm{O}(a|q)$.

In Equation \eqref{eq:linear1_ns} all the non-signalling constraints are listed, up to a factor of $Q(q)$ which does not change the constraints when considering complete-support games, but will be important later in the following section. We note that the only place in the proof where the complete-support property of the game is used is for writing down the linear program above. In Appendix~\ref{app:extension} we explain the implications of the linear program \eqref{eq:linear_non_relaxed} to games with incomplete-support. In particular, in Appendix~\ref{app:two_player_extension} we show how to modify program \eqref{eq:linear_non_relaxed} for the case of two-player games with incomplete-support such that our result still holds. In Appendix \ref{app:general_extension} we show how one can slightly modify the parallel repetition procedure to derive a general (although modified) threshold theorem for any game.

Next, one can relax the linear program \eqref{eq:linear_non_relaxed} to the following:
\begin{subequations} \label{eq:linear_program}
\begin{align}
	\max \quad &\sum_{q,a} Q(q) R(q,a) \mathrm{O}(a|q)\nonumber \\
	\text{s.t.} \quad&   Q(q^i , q^{\bar{i}})\left[ \mathrm{O}(\circ,a^{\bar{i}}|q^i , q^{\bar{i}}) - \sum_{r^{i}} Q(r^{i}|q^{\bar{i}})\mathrm{O}(\circ,a^{\bar{i}} | r^{i},q^{\bar{i}})\right] \leq 0  &\forall i, q^i,q^{\bar{i}},a^{\bar{i}} \label{eq:relaxed} \\
	&\sum_{a} \mathrm{O}(a|q) = 1  &\forall q  \nonumber\\
	& \mathrm{O}(a|q) \geq 0 &\forall a,q \nonumber
\end{align}
\end{subequations}

To see that the relaxation of the non-signalling constraints \eqref{eq:linear1_ns} to the constraints \eqref{eq:relaxed} does not change the program, i.e., does not change the value of the optimal solution, assume there exists $i, q^i,q^{\bar{i}},a^{\bar{i}}$ for which
\[
	Q(q^i , q^{\bar{i}})\left[ \mathrm{O}(\circ,a^{\bar{i}}|q^i , q^{\bar{i}}) - \sum_{r^{i}} Q(r^{i}|q^{\bar{i}})\mathrm{O}(\circ,a^{\bar{i}} | r^{i},q^{\bar{i}})\right] < 0 \;.
\]
That is, $\mathrm{O}(\circ,a^{\bar{i}}|q^i , q^{\bar{i}})$ is smaller than the average $\sum_{r^{i}} Q(r^{i}|q^{\bar{i}})\mathrm{O}(\circ,a^{\bar{i}} | r^{i},q^{\bar{i}})$, and therefore there must be some $s^i$ for which $\mathrm{O}(\circ,a^{\bar{i}}|s^i , q^{\bar{i}})$ is larger than the average. Meaning, 
\[
	Q(s^i , q^{\bar{i}})\left[ \mathrm{O}(\circ,a^{\bar{i}}|s^i , q^{\bar{i}}) - \sum_{r^{i}} Q(r^{i}|q^{\bar{i}})\mathrm{O}(\circ,a^{\bar{i}} | r^{i},q^{\bar{i}})\right] > 0 \;,
\]
but this contradicts the constraints in \eqref{eq:relaxed}.  

The dual program of the primal \eqref{eq:linear_program} is given below. 
\begin{subequations}\label{eq:dual}
\begin{align}
	\min \quad & \sum_q z(q) \nonumber \\
	\text{s.t.} \quad & z(q) + \sum_i y_i(q,a^{\bar{i}})Q(q) - \sum_i \sum_{\substack{r |\\ {r^{\bar{i}}}=q^{\bar{i}}}} y_i(r,a^{\bar{i}}) Q(r) Q(r^{i}|q^{\bar{i}}) \geq Q(q)R(q,a) & \forall a,q \label{eq:dual_constraint} \\
	& y_i(q,a^{\bar{i}}) \geq 0 &\forall i,q,a^{\bar{i}} \nonumber
\end{align}
\end{subequations}

\subsection{Signalling measure}

Given a general strategy $\mathrm{O}_{A|Q}$ we would like to measure the amount of signalling from every player $i\in [m]$ to all the other players together. Following the linear program \eqref{eq:linear_program}, we quantify signalling using Definition \ref{def:signalling} below. 

In the definition we derive all the relevant conditional and marginal distributions from $\mathrm{O}_{AQ}$. Concretely we use the following notation: $\mathrm{O}(\circ,b^{\bar{i}}|s^i , s^{\bar{i}})=\sum_{b^i}\mathrm{O}(b^i,b^{\bar{i}}|s^i , s^{\bar{i}})$ as before, $\mathrm{O}(\circ,b^{\bar{i}},\circ, s^{\bar{i}})=\sum_{b^i,s^i}\mathrm{O}(b^i,b^{\bar{i}},s^i , s^{\bar{i}})$, and 
\[
	\mathrm{O}(\circ, s^i | b^{\bar{i}}, s^{\bar{i}})=\sum_{b^i} \mathrm{O}(b^i, s^i | b^{\bar{i}}, s^{\bar{i}}) = \sum_{b^i} \frac{\mathrm{O}(b^i, b^{\bar{i}}, s^i, s^{\bar{i}})}{\mathrm{O}(\circ, b^{\bar{i}}, \circ, s^{\bar{i}})} \;.
\]

\begin{defn}[Signalling measure]\label{def:signalling}
	The signalling of strategy $\mathrm{O}_{A|Q}$ in direction $i \rightarrow \bar{i}$ using outputs $b^{\bar{i}}$ and inputs $s^i , s^{\bar{i}}$ is defined as
	\begin{align}
		\mathrm{Sig}_{(i, b^{\bar{i}}, s^i , s^{\bar{i}} )}\left( \mathrm{O}\right) &= Q(s^i , s^{\bar{i}})\left[ \mathrm{O}(\circ,b^{\bar{i}}|s^i , s^{\bar{i}}) - \sum_{r^{i}} Q(r^{i}|s^{\bar{i}})\mathrm{O}(\circ,b^{\bar{i}} | r^{i},s^{\bar{i}})\right] \label{eq:sig1}\\
		&= \mathrm{O}(\circ,b^{\bar{i}},\circ, s^{\bar{i}}) \left[ \mathrm{O}(\circ, s^i | b^{\bar{i}}, s^{\bar{i}}) - Q(s^i |s^{\bar{i}}) \right] \label{eq:sig2} \;.
	\end{align}
\end{defn}

That is, we have a signalling measure for every $(i, b^{\bar{i}}, s^i , s^{\bar{i}} )$. If $\mathrm{Sig}_{(i, b^{\bar{i}}, s^i , s^{\bar{i}} )}\left( \mathrm{O}\right)>0$ we say that  the strategy is signalling in direction $(i, b^{\bar{i}}, s^i , s^{\bar{i}} )$. A negative signalling value, $\mathrm{Sig}_{(i, b^{\bar{i}}, s^i , s^{\bar{i}} )}\left( \mathrm{O}\right)<0$, is not relevant due to the inequality in Equation \eqref{eq:relaxed}.

The two forms of $\mathrm{Sig}_{(i, b^{\bar{i}}, s^i , s^{\bar{i}} )}\left( \mathrm{O}\right)$ given in equations \eqref{eq:sig1} and \eqref{eq:sig2} are equivalent according to Bayes' rule and they will be useful in different places in the proof. Equation \eqref{eq:sig2} for example allows us to quantify the amount by which input $s^i$ is more or less probable given $b^{\bar{i}}$, compared to the prior $Q(s^i |s^{\bar{i}})$. 

The following lemma shows that our measure of signalling is continuous. That is, if two strategies are close to one another according to Definition \ref{def:trace_distance} then their signalling values are also close. The proof is given in Appendix \ref{app:signalling_proofs}. 

\begin{lemma}[Continuity of $\mathrm{Sig}$]\label{lem:continuity}
	Let $\mathrm{O}_1$ and $\mathrm{O}_2$ be two one-game strategies such that $\big|\mathrm{O}_1-\mathrm{O}_2\big|_1 \leq \epsilon$. Then 
	\[
		\forall i, b^{\bar{i}}, s^i , s^{\bar{i}} \quad  \big|\mathrm{Sig}_{(i, b^{\bar{i}}, s^i , s^{\bar{i}} )}\left( \mathrm{O}_1 \right) - \mathrm{Sig}_{(i, b^{\bar{i}}, s^i , s^{\bar{i}} )}\left( \mathrm{O}_2 \right)\big|\leq 2\epsilon \;.
	\]
\end{lemma}

\subsection{Signalling tests}

In the following we will need an operational way of testing whether a one-game strategy $\mathrm{O}_{A|Q}$ is signalling. This can be done by using many copies of $\mathrm{O}_{A|Q}$ -- given data $\vec{q},\vec{a}$ which is distributed according to many independent copies of $\mathrm{O}_{AQ}$ it is possible to create an estimation of $\mathrm{O}_{A|Q}$,  $\mathrm{O}^{\mathrm{EST1}}_{A|Q}$, and then evaluate $\mathrm{Sig}_{(i, b^{\bar{i}}, s^i , s^{\bar{i}} )}\left( \mathrm{O}^{\mathrm{EST1}}\right)$. 

To formulate this process we first define an indicator function which will be used in the test. More precisely, for every tuple $(i,b^{\bar{i}}, s^i , s^{\bar{i}})$ we define a function $\mathcal{T}_{(i,b^{\bar{i}}, s^i , s^{\bar{i}})}:\mathcal{Q}^t\times\mathcal{A}^t\rightarrow\{0,1\}$:

\begin{equation}\label{eq:test_definition}
	\mathcal{T}_{(i,b^{\bar{i}}, s^i , s^{\bar{i}})} (\vec{q^t},\vec{a^t})= 
	\begin{cases}
    		1& \text{if } \mathrm{Sig}_{(i, b^{\bar{i}}, s^i , s^{\bar{i}} )}\left( \mathrm{O}^{\mathrm{EST1}}\right) \geq \zeta - 2\epsilon\\
    		0& \text{otherwise}
	\end{cases} \\
\end{equation}
where $\mathrm{O}^{\mathrm{EST1}}$ is estimated from $\vec{q^t},\vec{a^t}$ and $\zeta,\epsilon>0$ are parameters satisfying $\zeta \geq 7\epsilon$ and $\epsilon \leq \min_q Q(q)$.
See Figure \ref{fig:sig_ways} for a visualisation of the different signalling forms $(i,b^{\bar{i}}, s^i , s^{\bar{i}})$ and the signalling values considered in the test.

\begin{figure}
\begin{centering}
	\begin{tikzpicture}
	
		\draw (0,0) node {$\mathrm{Sig}_{(i, b^{\bar{i}}, s^i , s^{\bar{i}} )}\left( \mathrm{O}\right)$};
		\draw  [<->] (2,0) -- (12,0);
		
		\draw[fill] (3,0) circle [radius=0.025];
		\node [below] at (3,0) {0};
		
		\draw[fill] (6,0) circle [radius=0.025];
		\node [below] at (6,0) {$\zeta-2\epsilon$};
		
		\draw[fill] (7.5,0) circle [radius=0.025];
		\node [below] at (7.5,0) {$\zeta$};

		\draw[decorate,decoration={brace,amplitude=10pt}] (6.10,0.25) -- (12,0.25) node [midway,yshift=0.75cm] {$\mathrm{Sig}_{(i, b^{\bar{i}}, s^i , s^{\bar{i}} )}\left( \mathrm{O}\right)\geq \zeta - 2\epsilon$};

	\end{tikzpicture}
\par\end{centering}
\caption{The different forms of signalling: every $i$ and every $b^{\bar{i}}, s^i , s^{\bar{i}}$ define a line as in the figure. The value of $\mathrm{Sig}_{(i, b^{\bar{i}}, s^i , s^{\bar{i}} )}\left( \mathrm{O}\right)$ tells us exactly where we are on the line.} \label{fig:sig_ways}
\end{figure}

The following observation will be crucial later on:
\begin{remark}\label{rem:locality}
	According to Definition \ref{def:signalling}, in order to evaluate $\mathrm{Sig}_{(i, b^{\bar{i}}, s^i , s^{\bar{i}} )}\left( \mathrm{O}^{\mathrm{EST1}} \right)$ there is no need to know $\mathrm{O}^{\mathrm{EST1}}$ completely; only the marginals $\mathrm{O}^{\mathrm{EST1}}(\circ,b^{\bar{i}}|r^{i},s^{\bar{i}})$ for all $r^{i}$ are needed.  
\end{remark}

For every $(i,b^{\bar{i}}, s^i , s^{\bar{i}})$ we can now consider a signalling test. Given a strategy $\mathrm{P}_{\vec{A}|\vec{Q}}$ for the repeated game $G^{n}$ we sample $n$ tuples of questions $\vec{q}$ using the game distribution $Q^{\otimes n}$ and use them to get $n$ tuples of answers $\vec{a}$ which are distributed according to $\mathrm{P}_{\vec{A}|\vec{Q}}$. Finally, if $\mathcal{T}_{(i,b^{\bar{i}}, s^i , s^{\bar{i}})}(\vec{q^t},\vec{a^t})=1$ the test accepts, and otherwise rejects\footnote{As $\mathrm{P}_{\vec{A}|\vec{Q}}$ can be signalling between the different $n$ tuples of questions-answers one has to input all the questions before getting the test answers.}. 
Note that if a question $s$ does not appear in the test data $\vec{q^t}$ the test $\mathcal{T}_{(i,b^{\bar{i}}, s^i , s^{\bar{i}})}$ rejects by definition.

The following lemma shows that the test is reliable when applied to an i.i.d.\ strategy $\mathrm{O}_{A|Q}^{\otimes n}$. That is, if $\mathrm{Sig}_{(i, b^{\bar{i}}, s^i , s^{\bar{i}} )}\left( \mathrm{O}\right)\geq \zeta$ the test will detect it with high probability, i.e. the test will accept with high probability, and if $\mathrm{O}_{A|Q}$ is non-signalling then the test will reject with high probability. The proof can be found in Appendix \ref{app:signalling_proofs}.

\begin{lemma}[Reliable signalling test]\label{lem:reliable_test}
	Assume the players share an i.i.d.\ strategy $\mathrm{O}_{A|Q}^{\otimes n}$. For every $(i, b^{\bar{i}}, s^i , s^{\bar{i}} )$,
	\begin{enumerate}
		\item If $\mathrm{Sig}_{(i, b^{\bar{i}}, s^i , s^{\bar{i}} )}\left( \mathrm{O}\right)\geq \zeta$ then 
			$\mathrm{Pr}_{\vec{a},\vec{q}\sim\mathrm{O}^{\otimes n}_{AQ}} \left[ \mathcal{T}_{(i,b^{\bar{i}}, s^i , s^{\bar{i}})}\left( \vec{q^t},\vec{a^t} \right) = 1 \right] > 1- \delta$.
		\item If $\mathrm{Sig}_{(i, b^{\bar{i}}, s^i , s^{\bar{i}} )}\left( \mathrm{O}\right)= 0$ then 
			$\mathrm{Pr}_{\vec{a},\vec{q}\sim\mathrm{O}^{\otimes n}_{AQ}} \left[ \mathcal{T}_{(i,b^{\bar{i}}, s^i , s^{\bar{i}})}\left( \vec{q^t},\vec{a^t} \right) = 0 \right] > 1 - \delta$.
	\end{enumerate}
	where $\delta=\delta\left(\frac{n}{2},\epsilon\right)=\left(\frac{n}{2}+1\right)^{|\mathcal{A}|\cdot|\mathcal{Q}|-1}e^{-n\epsilon^2/4}$.
\end{lemma}

Given a specific signalling test $\mathcal{T}_{(i,b^{\bar{i}}, s^i , s^{\bar{i}})}$ we define two relevant sets of one-game strategies:

\begin{align}
	&\sigma_{(i,b^{\bar{i}}, s^i , s^{\bar{i}})} = \left\{ \mathrm{O} \big| \forall \bar{\mathrm{O}} \text{ s.t. } |\mathrm{O} - \bar{\mathrm{O}}|_1 \leq \epsilon \;, \bar{\mathrm{O}} \text{ is $\zeta$ signalling or more in $(i,b^{\bar{i}}, s^i , s^{\bar{i}})$} \right\} \label{eq:sigma_def}\\
	&\Sigma_{(i,b^{\bar{i}}, s^i , s^{\bar{i}})} = \left\{ \mathrm{O} \big| \exists \bar{\mathrm{O}} \text{ s.t. } |\mathrm{O} - \bar{\mathrm{O}}|_1 \leq \epsilon \land  \mathrm{Pr}_{\vec{a},\vec{q}\sim\bar{\mathrm{O}}^{\otimes n}_{AQ}} \left[ \mathcal{T}_{(i,b^{\bar{i}}, s^i , s^{\bar{i}})}\left( \vec{q^t},\vec{a^t} \right) = 1 \right] > \delta \right\}
\end{align}

The following two lemmas allow us to address these sets also according to the signalling values of the relevant strategies (see also Figure~\ref{fig:sig_sets}).

\begin{lemma}\label{lem:small_sigma_signalling}
	For all $\mathrm{O}\notin\sigma_{(i,b^{\bar{i}}, s^i , s^{\bar{i}})}$, $\mathrm{Sig}_{(i, b^{\bar{i}}, s^i , s^{\bar{i}} )}\left( \mathrm{O} \right)<\zeta+2\epsilon$.
\end{lemma}

\begin{lemma}\label{lem:Sigma_signalling}
	Let $\nu > 0$ be any parameter such that $\nu < \zeta - 6\epsilon$. Then  
	\[
	 	\forall \mathrm{O}\in\Sigma_{(i,b^{\bar{i}}, s^i , s^{\bar{i}})}, \quad \mathrm{Sig}_{(i, b^{\bar{i}}, s^i , s^{\bar{i}} )}\left( \mathrm{O}\right) > \nu \;.
	\]
\end{lemma} 

Lemma \ref{lem:small_sigma_signalling} follows right away from Lemma \ref{lem:continuity} and the definition of $\sigma_{(i,b^{\bar{i}}, s^i , s^{\bar{i}})}$. Lemma \ref{lem:Sigma_signalling} is proven in Appendix \ref{app:signalling_proofs}.

\begin{figure}
\begin{centering}
	\begin{tikzpicture}
	
		\draw (0,0) node {$\mathrm{Sig}_{(i, b^{\bar{i}}, s^i , s^{\bar{i}} )}\left( \mathrm{O}\right)$};
		\draw  [<->] (2,0) -- (12,0);
		
		\draw[fill] (3,0) circle [radius=0.025];
		\node [below] at (3,0) {0};
		
		\draw[fill] (4.5,0) circle [radius=0.025];
		\node [below] at (4.5,0) {$\nu$};
		
		\draw[fill] (7.5,0) circle [radius=0.025];
		\node [below] at (7.5,0) {$\zeta$};
		
		\draw[fill] (9,0) circle [radius=0.025];
		\node [below] at (9,0) {$\zeta+2\epsilon$};

		\draw[decorate,decoration={brace,mirror,amplitude=10pt}] (4.6,-0.5) -- (12,-0.5) node [midway,yshift=-0.75cm] {$\mathrm{Sig}$ of $\mathrm{O}\in\Sigma_{(i,b^{\bar{i}}, s^i , s^{\bar{i}})}$};
		\draw[decorate,decoration={brace,amplitude=10pt}] (2.10,0.25) -- (8.90,0.25) node [midway,yshift=0.75cm] {$\mathrm{Sig}$ of $\mathrm{O}\notin\sigma_{(i,b^{\bar{i}}, s^i , s^{\bar{i}})}$};
		\draw[decorate,decoration={brace,mirror,amplitude=10pt}] (3.10,-0.5) -- (4.5,-0.5) node [midway,yshift=-0.75cm] {constant gap};
	
	\end{tikzpicture}
\par\end{centering}
\caption{Visualization of the signalling values which are relevant for Lemma \ref{lem:Sigma_signalling} and the sets $\sigma_{(i,b^{\bar{i}}, s^i , s^{\bar{i}})},\Sigma_{(i,b^{\bar{i}}, s^i , s^{\bar{i}})}$.} \label{fig:sig_sets}
\end{figure}

\section{Using de Finetti strategies}\label{sec:deFinetti}

In this section we start analysing the relation between the test questions-answers $\vec{q^t},\vec{a^t}$ and the game questions-answers $\vec{q^g},\vec{a^g}$ in one instance of the repeated game $G^{n}$ using a strategy $\mathrm{P}_{\vec{A}|\vec{Q}}$. More precisely, we denote the one-game strategy which is estimated from  $\vec{q^g},\vec{a^g}$ by $\mathrm{O}^{\mathrm{EST2}}_{A|Q}$, and we are interested in knowing what is the probability that $\mathrm{O}^{\mathrm{EST2}}_{A|Q}\in \Sigma_{(i,b^{\bar{i}}, s^i , s^{\bar{i}})}$ or $\mathrm{O}^{\mathrm{EST2}}_{A|Q}\in \sigma_{(i,b^{\bar{i}}, s^i , s^{\bar{i}})}$ given the result of $\mathcal{T}_{(i,b^{\bar{i}}, s^i , s^{\bar{i}})} (\vec{q^t},\vec{a^t})$.

We first do this for any i.i.d.\ strategy and then extend the results to any permutation-invariant strategy using a de Finetti reduction \cite{arnon2013finetti}.

\subsection{de Finetti strategies}

As mentioned in Section \ref{sec:intro}, de Finetti strategies are strategies that can be written as a convex combination of i.i.d.\ strategies. Formally,
\begin{defn}[de Finetti strategy]
	A de Finetti strategy $\tau_{\vec{A}|\vec{Q}}$ is a strategy of the form
	\[
		\tau_{\vec{A}|\vec{Q}} = \int \mathrm{O}_{A|Q}^{\otimes n} \mathrm{d}\mathrm{O}_{A|Q}\;,
	\]
	where $\mathrm{d}\mathrm{O}_{A|Q}$ is some measure on the space of one-game strategies.
\end{defn}

In the following lemma we are interested in the relation between the test questions-answers $\vec{q^t},\vec{a^t}$ and the game questions-answers $\vec{q^g},\vec{a^g}$ in one instance of the repeated game $G^{n}$. For i.i.d.\ strategies (and therefore also for de Finetti strategies) this is simple: $\vec{q^t},\vec{a^t}$ and $\vec{q^g},\vec{a^g}$ are independent of each other and conditioning on a property of one of them does not affect the other.

\begin{lemma} \label{lem:de_finetti_prop}
	For a de Finetti strategy $\tau_{\vec{A}|\vec{Q}}$ and every $(i,b^{\bar{i}}, s^i , s^{\bar{i}})$
	\begin{enumerate}
		\item $\mathrm{Pr}_{\vec{a},\vec{q}\sim\tau_{\vec{A}\vec{Q}}} \left[\mathcal{T}_{(i,b^{\bar{i}}, s^i , s^{\bar{i}})} (\vec{q^t},\vec{a^t})=1 \land \mathrm{O}^{\mathrm{EST2}}_{A|Q} \notin \Sigma_{(i,b^{\bar{i}}, s^i , s^{\bar{i}})}  \right] \leq 2\delta $
		\item $\mathrm{Pr}_{\vec{a},\vec{q}\sim\tau_{\vec{A}\vec{Q}}} \left[ \mathcal{T}_{(i,b^{\bar{i}}, s^i , s^{\bar{i}})} (\vec{q^t},\vec{a^t})=0 \land \mathrm{O}^{\mathrm{EST2}}_{A|Q} \in \sigma_{(i,b^{\bar{i}}, s^i , s^{\bar{i}})}  \right] \leq 2\delta $
	\end{enumerate}
\end{lemma}

The proof of this lemma (given in Appendix \ref{app:proofs_de_finetti}) follows from Sanov's theorem stated in Lemma~\ref{lem:chernoff}. Intuitively, if $\mathcal{T}_{(i,b^{\bar{i}}, s^i , s^{\bar{i}})} (\vec{q^t},\vec{a^t})=1$ then $\mathrm{O}^{\mathrm{EST1}}$ is signalling and therefore so should $\mathrm{O}^{\mathrm{EST2}}$ be, and vice versa. 

\subsection{de Finetti reductions}

Of course, considering just de Finetti strategies is not interesting by itself. Luckily, we can now use a de Finetti reduction to extend the results of the previous section to any permutation-invariant strategy, where the permutation is performed on the questions-answers pairs (we do not permute the players). As the repeated game $G^{n}$ is by itself permutation invariant we can restrict the strategies of the players to be permutation invariant without loss of generality. 

\begin{defn}[Permutation invariance]\label{def:permutation}
	Given a strategy $\mathrm{P}_{\vec{A}|\vec{Q}}$ and a permutation $\pi$ of the questions and answers we denote by $\mathrm{P}_{\vec{A}|\vec{Q}}\circ\pi$ the strategy which is defined by 
	\[
		\forall \vec{a},\vec{q} \quad \left(\mathrm{P}_{\vec{A}|\vec{Q}}\circ\pi \right) (\vec{a}|\vec{q})=\mathrm{P}_{\vec{A}|\vec{Q}}(\pi(\vec{a})|\pi(\vec{q})) \;.
	\]
	$\mathrm{P}_{\vec{A}|\vec{Q}}$ is permutation invariant if for any permutation $\pi$, $\mathrm{P}_{\vec{A}|\vec{Q}} = \mathrm{P}_{\vec{A}|\vec{Q}}\circ\pi$.  
\end{defn}

The following lemma shows that we can restrict our analysis to permutation-invariant strategies without loss of generality. 

\begin{lemma}\label{lem:prem_wlog}
	For every strategy $\mathrm{P}_{\vec{A}|\vec{Q}}$ for the repeated game $G^{n}$ there exists a permutation-invariant strategy $\tilde{\mathrm{P}}_{\vec{A}|\vec{Q}}$ such that $w\left( \mathrm{P}_{\vec{A}|\vec{Q}}\right) =w\left( \tilde{\mathrm{P}}_{\vec{A}|\vec{Q}}\right)$.
\end{lemma}
\begin{proof}
	Given $\mathrm{P}_{\vec{A}|\vec{Q}}$ define its permutation-invariant version to be 
	\[
		\tilde{\mathrm{P}}_{\vec{A}|\vec{Q}} = \frac{1}{n!} \sum_{\pi}  \mathrm{P}_{\vec{A}|\vec{Q}}\circ\pi \;.
	\]
	The winning probability of the game is linear in the strategy, therefore we have 
	\begin{equation}\label{eq:winning_prob_linear}
		w\left( \tilde{\mathrm{P}}_{\vec{A}|\vec{Q}} \right) = w\left( \frac{1}{n!} \sum_{\pi}  \mathrm{P}_{\vec{A}|\vec{Q}}\circ\pi \right) =  \frac{1}{n!} \sum_{\pi} w\left( \mathrm{P}_{\vec{A}|\vec{Q}}\circ\pi \right) \;.
	\end{equation}
	
	Since the tuples of questions in the repeated game are chosen in an i.i.d.\ manner and the winning condition is checked for each tuple separately, the winning probability is indifferent to the ordering of the questions-answers pairs. As $\pi$ permutes the tuples of questions and answers together we have $w\left( \mathrm{P}_{\vec{A}|\vec{Q}}\circ\pi \right) = w\left( \mathrm{P}_{\vec{A}|\vec{Q}} \right)$.  
	
	Combining this with Equation \eqref{eq:winning_prob_linear} we get $w\left( \tilde{\mathrm{P}}_{\vec{A}|\vec{Q}} \right)=w\left( \mathrm{P}_{\vec{A}|\vec{Q}} \right)$. 
\end{proof}

\begin{lemma}[de Finetti reduction for conditional probability distributions \cite{arnon2013finetti}] \label{lem:deFinetti_reduction}
	Let $c = (n+1)^{|\mathcal{Q}|(|\mathcal{A}|-1)}$. There exists a de Finetti strategy $\tau_{\vec{A}|\vec{Q}}$ such that for every permutation-invariant strategy $\mathrm{P}_{\vec{A}|\vec{Q}}$  
	\[
		\forall \vec{a},\vec{q} \quad \mathrm{P}_{\vec{A}|\vec{Q}} (\vec{a}|\vec{q}) \leq c \cdot  \tau_{\vec{A}|\vec{Q}} (\vec{a}|\vec{q}) \;.
	\]
\end{lemma}
The de Finetti strategy $\tau_{\vec{A}|\vec{Q}}$ is constructed explicitly in \cite{arnon2013finetti} but the specific construction is not relevant for our purposes. In some special cases the constant $c$ in Lemma \ref{lem:deFinetti_reduction} can also be made smaller by taking into account symmetries of the game $G$ itself. For further details see \cite{arnon2013finetti}.

We now use the de Finetti reduction to show that the properties proven in Lemma \ref{lem:de_finetti_prop} for the de Finetti strategy also hold true for permutation-invariant strategies, although with slightly weaker parameters. Concretely, the bound of $2\delta$ in Lemma \ref{lem:de_finetti_prop} is replaced by $2c\delta$ in the following lemma. Nevertheless, the bound still decreases exponentially fast with the number of repetitions\footnote{One would have liked to apply a similar argument to the winning probability of the repeated game right away. That is, $w(\mathrm{P}_{\vec{A}|\vec{Q}})\leq c w(\tau_{\vec{A}|\vec{Q}})$. This claim is indeed correct, but not useful. A look at the explicit construction of $\tau_{\vec{A}|\vec{Q}}$ itself in \cite{arnon2013finetti} will reveal that it is a signalling strategy, hence no non-trivial bound on $w(\tau_{\vec{A}|\vec{Q}})$ holds a priori. For a further discussion see Section \ref{sec:what_we_learn}.}.
 
\begin{lemma}[Reduction]\label{lem:reduction}
	For every permutation-invariant strategy $\mathrm{P}_{\vec{A}|\vec{Q}}$ and every $(i,b^{\bar{i}}, s^i , s^{\bar{i}})$
	\begin{enumerate}
		\item \label{it:reduction1}$\mathrm{Pr}_{\vec{a},\vec{q}\sim\mathrm{P}_{\vec{A}\vec{Q}}} \left[ \mathcal{T}_{(i,b^{\bar{i}}, s^i , s^{\bar{i}})} (\vec{q^t},\vec{a^t})=1 \land \mathrm{O}^{\mathrm{EST2}}_{A|Q} \notin \Sigma_{(i,b^{\bar{i}}, s^i , s^{\bar{i}})}  \right] \leq 2c \delta$ .
		\item \label{it:reduction2} $\mathrm{Pr}_{\vec{a},\vec{q}\sim\mathrm{P}_{\vec{A}\vec{Q}}} \left[ \mathcal{T}_{(i,b^{\bar{i}}, s^i , s^{\bar{i}})} (\vec{q^t},\vec{a^t})=0 \land \mathrm{O}^{\mathrm{EST2}}_{A|Q} \in \sigma_{(i,b^{\bar{i}}, s^i , s^{\bar{i}})}  \right] \leq 2c \delta$ .
	\end{enumerate}
\end{lemma}

\begin{proof}
	We prove both of the claims together. Denote the relevant event by $E(\vec{a},\vec{q})$ and note that for both events we can write 
	\[
		\mathrm{Pr}_{\vec{a},\vec{q}\sim\mathrm{P}_{\vec{A}\vec{Q}}} \left[ E(\vec{a},\vec{q}) =1 \right] = \sum_{\substack{\vec{a},\vec{q} | \\ E(\vec{a},\vec{q})=1 }} \mathrm{P}_{\vec{A}\vec{Q}}(\vec{a},\vec{q}) \;.
	\]
	From Lemma \ref{lem:deFinetti_reduction} we get $\mathrm{P}_{\vec{A}\vec{Q}}(\vec{a},\vec{q}) \leq c \cdot \tau_{\vec{A}\vec{Q}}(\vec{a},\vec{q})$ and therefore 
	\[
		\mathrm{Pr}_{\vec{a},\vec{q}\sim\mathrm{P}_{\vec{A}\vec{Q}}} \left[ E(\vec{a},\vec{q}) =1 \right] = \sum_{\substack{\vec{a},\vec{q} | \\ E(\vec{a},\vec{q})=1 }} \mathrm{P}_{\vec{A}\vec{Q}}(\vec{a},\vec{q}) \leq c \cdot \sum_{\substack{\vec{a},\vec{q} | \\ E(\vec{a},\vec{q})=1 }} \tau_{\vec{A}\vec{Q}}(\vec{a},\vec{q}) = c \cdot \mathrm{Pr}_{\vec{a},\vec{q}\sim\tau_{\vec{A}\vec{Q}}} \left[ E(\vec{a},\vec{q}) =1 \right] \;.
	\]
	Combining this with Lemma \ref{lem:de_finetti_prop} proves the lemma. 
\end{proof}

\section{Threshold theorem}\label{sec:threshold}

In this section we prove our threshold theorem, Theorem \ref{thm:final_threshold_theorem}. Before going into the details of the proof, let us explain the high-level idea. 

First, to see the connection between what was done so far and a threshold theorem note that the winning probability of $\mathrm{O}^{\mathrm{EST2}}_{A|Q}$ in the game $G$, $w(\mathrm{O}^{\mathrm{EST2}}_{A|Q})$, is exactly the fraction of coordinates in which the game data $\vec{q^g},\vec{a^g}$ satisfies the winning condition $R$. Therefore, in order to prove a threshold theorem it is sufficient to prove an upper bound on $w(\mathrm{O}^{\mathrm{EST2}}_{A|Q})$ which holds with high probability. 

To do so we use the following sequence of lemmas. The first two lemmas bound the probability that the estimate $\mathrm{O}^{\mathrm{EST2}}_{A|Q}$ is significantly signalling\footnote{In the words of the explanation given in Section \ref{sec:proof_idea}, this is where we prove that the signalling weight is exponentially small.} in any direction $(i,b^{\bar{i}}, s^i , s^{\bar{i}})$ for which $\mathrm{Pr}_{\vec{a},\vec{q}\sim\mathrm{P}_{\vec{A}\vec{Q}}} \big[\mathcal{T}_{(i,b^{\bar{i}}, s^i , s^{\bar{i}})}( \vec{q^t},\vec{a^t} )=1 \big] \neq 0$. Lemma~\ref{lem:weak_condition}, which we also call the weak lemma, establishes that even conditioned on the test $\mathcal{T}_{(i,b^{\bar{i}}, s^i , s^{\bar{i}})}( \vec{q^t},\vec{a^t})$ detecting signalling the distribution  $\mathrm{O}^{\mathrm{EST2}}_{A|Q}$ itself cannot be signalling with very high probability. The proof of the lemma is based on a reduction to a certain guessing game which is used to derive a contradiction between the conclusion that $\mathrm{O}^{\mathrm{EST2}}_{A|Q}$ would be signalling and the assumption that the overall distribution $\mathrm{P}_{\vec{A}|\vec{Q}}$ is not. Lemma \ref{lem:strong_condition}, called the strong lemma, amplifies the conclusion of the weak lemma to show that $\mathrm{O}^{\mathrm{EST2}}_{A|Q}$ cannot display too much signalling, even only with small probability. The amplification is obtained by using the properties of permutation-invariant strategies which were proven in Lemma \ref{lem:reduction} in the previous section.  

Having shown that with high probability $\mathrm{O}^{\mathrm{EST2}}_{A|Q}$ cannot be too signalling, Lemma~\ref{lem:distance_transformation} derives an upper bound on the  winning probability $w(\mathrm{O}^{\mathrm{EST2}}_{A|Q})$. Intuitively, if the strategy $\mathrm{O}^{\mathrm{EST2}}_{A|Q}$ does not display strong signalling in any direction it should not lead to a large advantage over strictly non-signalling strategies in the game $G$. The  quantitative argument is based on performing a sensitivity analysis of the appropriate linear program. The three lemmas are brought together in~Lemma \ref{lem:threshold}, from which Theorem~\ref{thm:final_threshold_theorem} follows. 

We are now ready to prove the following lemmas and the threshold theorem.

\begin{lemma}[Weak lemma]\label{lem:weak_condition}
	Let $n$ be such that 
	\begin{equation}\label{eq:number_of_repet_bound}
		\frac{n}{\ln(n)} > 20|Q||A| \frac{\ln(2/\epsilon)}{\epsilon^{2}}\;,
	\end{equation} 
	and $\mathrm{P}_{\vec{A}|\vec{Q}}$ a non-signalling strategy for $G^{n}$. For any $(i,b^{\bar{i}}, s^i , s^{\bar{i}})$ denote by $\mathrm{P}_{\vec{A}\vec{Q}|\mathcal{T}=1}$ the probability distribution $\mathrm{P}_{\vec{A}\vec{Q}}$ conditioned on the event $\mathcal{T}_{(i,b^{\bar{i}}, s^i , s^{\bar{i}})}\left( \vec{q^t},\vec{a^t} \right)=1$, whenever such a conditional probability distribution is defined. Then,
	\begin{equation}\label{eq:weak_lemma_eq}
		\mathrm{Pr}_{\vec{a^g},\vec{q^g}\sim\mathrm{P}_{\vec{A}\vec{Q}|\mathcal{T}=1}} \left[\mathrm{O}^{\mathrm{EST2}}_{A|Q} \in \Sigma_{(i,b^{\bar{i}}, s^i , s^{\bar{i}})} \right]  < 1 - 2c \delta\;.
	\end{equation}
\end{lemma}
\begin{proof}

	For every signalling test $\mathcal{T}_{(i,b^{\bar{i}}, s^i , s^{\bar{i}})}$ we construct a guessing game. Our goal is to derive a contradiction by showing that if Equation \eqref{eq:weak_lemma_eq} is not true, then the guessing game can be won with probability higher than the optimal non-signalling winning probability. 
		
	The guessing game is defined as follows. A referee gives the players $n/2$ independent $m$-tuples of game questions $\vec{q^g}$ where each tuple is distributed according to the questions distribution $Q$ of the original game $G$. Players $\bar{i}$ are then allowed to communicate and their goal is to guess and output an index $j\in[n/2]$ such that $\vec{q^g}_j = (s^i , s^{\bar{i}})$ (if there is no such index the players lose automatically). 

	If the players share a non-signalling strategy $\mathrm{P}_{\vec{A}|\vec{Q}}$ then the marginals of players $\bar{i}$ are the same for all $q^i$. Therefore, their outputs $a^{\bar{i}}$ do not give them any information about the question that the $i$'th player got from the referee (even when players $\bar{i}$ are allowed to communicate among themselves, but not with player $i$). The best non-signalling strategy is therefore to choose, uniformly at random, an index $j$ for which ${\vec{q^g}_j}^{\bar{i}}=s^{\bar{i}}$ if it exists. The winning probability is then given by $W_{\mathrm{ns}} = Q(s^i|s^{\bar{i}})<1$.

	We now show that if the players share $\mathrm{P}_{\vec{A}|\vec{Q}}$ for which 
	\begin{equation}\label{eq:contr_assump}
		\mathrm{Pr}_{\vec{a^g},\vec{q^g}\sim\mathrm{P}_{\vec{A}\vec{Q}|\mathcal{T}=1}} \left[\mathrm{O}^{\mathrm{EST2}}_{A|Q} \in \Sigma_{(i,b^{\bar{i}}, s^i , s^{\bar{i}})} \right]  \geq 1 - 2c\delta 
	\end{equation}
	then they can win the above guessing game with probability higher than the optimal non-signalling winning probability~$W_{\mathrm{ns}}$.
	
	The general idea is as follows. The players use the questions given by the referee as the game questions $\vec{q^g}$ and choose, using shared randomness, the inputs for the test questions $\vec{q^t}$. They input the questions into $\mathrm{P}_{\vec{A}|\vec{Q}}$. Players $\bar{i}$ then check if $\mathcal{T}_{(i,b^{\bar{i}}, s^i , s^{\bar{i}})} (\vec{q^t},\vec{a^t})=1$ -- they can do this since they are allowed to communicate among themselves and they know all the inputs for the test questions of player $i$ (as they were chosen using shared randomness which is available to all the players). Recalling Remark \ref{rem:locality}, they have all the information they need. 
	
	The players proceed according to the following conditions:
	\begin{enumerate}
		\item If $\mathcal{T}_{(i,b^{\bar{i}}, s^i , s^{\bar{i}})} (\vec{q^t},\vec{a^t})=0$ they use the non-signalling strategy described above. That is, they choose a random index $j\in[n/2]$ such that ${\vec{q^g}_j}^{\bar{i}}=s^{\bar{i}}$.
		\item If $\mathcal{T}_{(i,b^{\bar{i}}, s^i , s^{\bar{i}})} (\vec{q^t},\vec{a^t})=1$ they choose a random index $j\in[n/2]$ such that ${\vec{q^g}_j}^{\bar{i}}=s^{\bar{i}}$ and ${\vec{a^g}_j}^{\bar{i}}=b^{\bar{i}}$ if it exists (otherwise they use the non-signalling strategy described above). 
	\end{enumerate}
	
	Let us show that, as long as $\mathrm{Pr}_{\vec{a},\vec{q}\sim\mathrm{P}_{\vec{A}\vec{Q}}} \big[\mathcal{T}_{(i,b^{\bar{i}}, s^i , s^{\bar{i}})}\left( \vec{q^t},\vec{a^t} \big)=1 \right] \neq 0$, this strategy achieves a winning probability which is higher than $W_{\mathrm{ns}}$. If $\mathcal{T}_{(i,b^{\bar{i}}, s^i , s^{\bar{i}})} (\vec{q^t},\vec{a^t})=0$ then the winning probability is $W_{\mathrm{ns}}$. However, if $\mathcal{T}_{(i,b^{\bar{i}}, s^i , s^{\bar{i}})} (\vec{q^t},\vec{a^t})=1$ then $\vec{q^g},\vec{a^g}$ can be seen as data which is distributed according to $n/2$ identical copies of $\mathrm{O}^{\mathrm{EST2}}$, which is with high probability in $\Sigma_{(i,b^{\bar{i}}, s^i , s^{\bar{i}})}$ according to Equation \eqref{eq:contr_assump}. From Lemma \ref{lem:Sigma_signalling} this implies 
	\begin{equation}\label{eq:weak_cond1_plus}
		\mathrm{Pr}_{\vec{a^g},\vec{q^g}\sim\mathrm{P}_{\vec{A}\vec{Q}|\mathcal{T}=1}} \left[\mathrm{Sig}_{(i,b^{\bar{i}}, s^i , s^{\bar{i}})}(\mathrm{O}^{\mathrm{EST2}}) > \nu \right]   \geq 1 - 2c \delta \;,
	\end{equation}
	where $\nu>0$ is any parameter satisfying $\nu < \zeta - 6\epsilon$ (recall Lemma \ref{lem:Sigma_signalling}).
	
	Using the definition of $\mathrm{Sig}$ in Equation \eqref{eq:sig2} we know that if indeed $\mathrm{Sig}_{(i,b^{\bar{i}}, s^i , s^{\bar{i}})}(\mathrm{O}^{\mathrm{EST2}}) > \nu $ then $\mathrm{O}^{\mathrm{EST2}}(\circ,b^{\bar{i}},\circ, s^{\bar{i}})>0$ and 
	\begin{align}
		 \mathrm{O}^{\mathrm{EST2}}(\circ, s^i | b^{\bar{i}}, s^{\bar{i}})  &> \frac{\nu}{\mathrm{O}^{\mathrm{EST2}}(\circ,b^{\bar{i}},\circ, s^{\bar{i}})} + Q(s^i |s^{\bar{i}}) \label{eq:weak_cond2_plus}  \\
		&= \frac{\nu}{\mathrm{O}^{\mathrm{EST2}}(\circ,b^{\bar{i}},\circ, s^{\bar{i}})} + W_{\mathrm{ns}} \;. \nonumber
	\end{align}
	That is, by choosing an index for which $a^{\bar{i}}=b^{\bar{i}}$ players $\bar{i}$ increase the winning probability. 
	
	On the other hand, if $\mathrm{Sig}_{(i,b^{\bar{i}}, s^i , s^{\bar{i}})}(\mathrm{O}^{\mathrm{EST2}}) \leq \nu $, which can happen with probability $ 2c\delta$, then the players might decrease their winning probability. In the worst case the winning probability is~0.
	
	Therefore, if $\mathcal{T}_{(i,b^{\bar{i}}, s^i , s^{\bar{i}})} (\vec{q^t},\vec{a^t})=1$ we get the following winning probability
	\begin{equation} \label{eq:win_cond_on_test}
		W_{|\mathcal{T}=1} \geq  (1 - 2c\delta) \left( \frac{\nu}{\mathrm{O}^{\mathrm{EST2}}(\circ,b^{\bar{i}},\circ, s^{\bar{i}})} + W_{\mathrm{ns}} \right) +  2c\delta \cdot 0
	\end{equation}
	
	and altogether, the winning probability of the described strategy is given by:
	
	\begin{align*}
		W &> \mathrm{Pr}_{\vec{a},\vec{q}\sim\mathrm{P}_{\vec{A}\vec{Q}}} \left[\mathcal{T}_{(i,b^{\bar{i}}, s^i , s^{\bar{i}})}\left( \vec{q^t},\vec{a^t} \right)=0 \right] W_{\mathrm{ns}}  \\ 
		&+  \mathrm{Pr}_{\vec{a},\vec{q}\sim\mathrm{P}_{\vec{A}\vec{Q}}} \left[\mathcal{T}_{(i,b^{\bar{i}}, s^i , s^{\bar{i}})}\left( \vec{q^t},\vec{a^t} \right)=1 \right] (1 - 2c\delta) \left( \frac{\nu}{\mathrm{O}^{\mathrm{EST2}}(\circ,b^{\bar{i}},\circ, s^{\bar{i}})} + W_{\mathrm{ns}} \right) \\
		&=  W_{\mathrm{ns}} - \mathrm{Pr}_{\vec{a},\vec{q}\sim\mathrm{P}_{\vec{A}\vec{Q}}} \left[\mathcal{T}_{(i,b^{\bar{i}}, s^i , s^{\bar{i}})}\left( \vec{q^t},\vec{a^t} \right)=1 \right] 2c\delta W_{\mathrm{ns}} \\
		&+  \mathrm{Pr}_{\vec{a},\vec{q}\sim\mathrm{P}_{\vec{A}\vec{Q}}} \left[\mathcal{T}_{(i,b^{\bar{i}}, s^i , s^{\bar{i}})}\left( \vec{q^t},\vec{a^t} \right)=1 \right] (1 - 2c\delta)  \frac{\nu}{\mathrm{O}^{\mathrm{EST2}}(\circ,b^{\bar{i}},\circ, s^{\bar{i}})} \;.
	\end{align*}

	Finally, $W>W_{\mathrm{ns}}$ for 
	\begin{equation}\label{eq:rep_num_condition}
		\nu > \frac{2c\delta}{1-2c\delta} W_{\mathrm{ns}} \geq \frac{2c\delta}{1-2c\delta} W_{\mathrm{ns}} \cdot \mathrm{O}^{\mathrm{EST2}}(\circ,b^{\bar{i}},\circ, s^{\bar{i}})  \;.
	\end{equation}
	Using $W_{\mathrm{ns}} \cdot \mathrm{O}^{\mathrm{EST2}}(\circ,b^{\bar{i}},\circ, s^{\bar{i}}) \leq 1$ and $2c\delta \leq (n+1)^{2|Q||A|}e^{-n\epsilon^2/4}$ (see Table~\ref{tb:parameters_table}), we see that as long as $n/\ln(n) > 20|Q||A| \epsilon^{-2}\ln(2/\epsilon)$ the quantity $2c\delta W_{\mathrm{ns}}  \mathrm{O}^{\mathrm{EST2}}(\circ,b^{\bar{i}},\circ, s^{\bar{i}}) /(1-2c\delta) $ is strictly less than~$\epsilon$. Assuming $\zeta \geq 7\epsilon$, there is a choice of $\nu$ that satisfies both~\eqref{eq:rep_num_condition} and the earlier condition that $\nu < \zeta - 6\epsilon$. 
\end{proof}

The bound given in Equation \eqref{eq:weak_lemma_eq} is weak for two reasons. First, the game data $\vec{q^g},\vec{a^g}$ is distributed according to the conditional distribution $\mathrm{P}_{\vec{A}\vec{Q}|\mathcal{T}=1}$ and not according to $\mathrm{P}_{\vec{A}\vec{Q}}$ itself. Second, it only tells us that 
$\mathrm{Pr}_{\vec{a^g},\vec{q^g}\sim\mathrm{P}_{\vec{A}\vec{Q}|\mathcal{T}=1}} \big[\mathrm{O}^{\mathrm{EST2}}_{A|Q} \notin \Sigma_{(i,b^{\bar{i}}, s^i , s^{\bar{i}})} \big]  \geq 2c\delta$, i.e., the probability that $\mathrm{O}^{\mathrm{EST2}}_{A|Q}$ has a small value of signalling is higher than $2c\delta$. We show how the statement in the weak lemma can be amplified using the de Finetti reduction from Lemma~\ref{lem:reduction}. 

\begin{lemma}[Strong lemma]\label{lem:strong_condition}
	Let $\mathrm{P}_{\vec{A}|\vec{Q}}$ be a permutation-invariant non-signalling strategy for $G^{n}$. Then for any $(i,b^{\bar{i}}, s^i , s^{\bar{i}})$ such that $Q(s^i , s^{\bar{i}})\neq 0$ and $Q(s^i|s^{\bar{i}})\neq 1$,
	\[
		\mathrm{Pr}_{\vec{a},\vec{q}\sim\mathrm{P}_{\vec{A}\vec{Q}}} \left[\mathrm{O}^{\mathrm{EST2}}_{A|Q} \in \sigma_{(i,b^{\bar{i}}, s^i , s^{\bar{i}})} \right] \leq 4c\delta\;.
	\]
\end{lemma}
\begin{proof}
	From Lemma \ref{lem:reduction} part \ref{it:reduction1} we get
	\[
		\mathrm{Pr}_{\vec{a},\vec{q}\sim\mathrm{P}_{\vec{A}\vec{Q}}} \left[\mathcal{T}_{(i,b^{\bar{i}}, s^i , s^{\bar{i}})}\left( \vec{q^t},\vec{a^t} \right)=1 \right] > 2c\delta \Rightarrow \mathrm{Pr}_{\vec{a},\vec{q} \sim \mathrm{P}_{\vec{A}\vec{Q}|\mathcal{T}=1}} \left[ \mathrm{O}^{\mathrm{EST2}}_{A|Q} \notin \Sigma_{(i,b^{\bar{i}}, s^i , s^{\bar{i}})}  \right] \leq 2c\delta
	\]
	and this can be rewritten as 
	\[
		\mathrm{Pr}_{\vec{a},\vec{q}\sim\mathrm{P}_{\vec{A}\vec{Q}}} \left[\mathcal{T}_{(i,b^{\bar{i}}, s^i , s^{\bar{i}})}\left( \vec{q^t},\vec{a^t} \right)=1 \right] > 2c\delta \Rightarrow \mathrm{Pr}_{\vec{a},\vec{q} \sim \mathrm{P}_{\vec{A}\vec{Q}|\mathcal{T}=1}} \left[ \mathrm{O}^{\mathrm{EST2}}_{A|Q} \in \Sigma_{(i,b^{\bar{i}}, s^i , s^{\bar{i}})}  \right]  \geq 1 - 2c\delta \;.
	\]
	According to Lemma \ref{lem:weak_condition}, this implies 
	\[
		\mathrm{Pr}_{\vec{a},\vec{q}\sim\mathrm{P}_{\vec{A}\vec{Q}}} \left[\mathcal{T}_{(i,b^{\bar{i}}, s^i , s^{\bar{i}})}\left( \vec{q^t},\vec{a^t} \right)=1 \right] > 2c\delta \Rightarrow \mathrm{P}_{\vec{A}|\vec{Q}} \text{ is signalling} \;.
	\]
	Therefore it must be that
	\begin{equation}\label{eq:test_prob1}
		\mathrm{Pr}_{\vec{a},\vec{q}\sim\mathrm{P}_{\vec{A}\vec{Q}}} \left[ \mathcal{T}_{(i,b^{\bar{i}}, s^i , s^{\bar{i}})}\left( \vec{q^t},\vec{a^t} \right)=1 \right] \leq 2c\delta\;
	\end{equation} 
	or alternatively, 
	\begin{equation}\label{eq:test_prob2}
		\mathrm{Pr}_{\vec{a},\vec{q}\sim\mathrm{P}_{\vec{A}\vec{Q}}} \left[ \mathcal{T}_{(i,b^{\bar{i}}, s^i , s^{\bar{i}})}\left( \vec{q^t},\vec{a^t} \right)=0 \right] \geq 1 - 2c\delta\;
	\end{equation} 
	
	Next, combining Lemma \ref{lem:reduction} part \ref{it:reduction2} with Equation \eqref{eq:test_prob2} we get 
	\[
		\mathrm{Pr}_{\vec{a},\vec{q} \sim \mathrm{P}_{AQ|\mathcal{T}=0}} \left[ \mathrm{O}^{\mathrm{EST2}}_{A|Q} \in \sigma_{(i,b^{\bar{i}}, s^i , s^{\bar{i}})} \right] \leq 2c\delta \;.
	\]
	Using Equation \eqref{eq:test_prob1} we get
	\[
		\mathrm{Pr}_{\vec{a},\vec{q} \sim \mathrm{P}_{\vec{A}\vec{Q}}} \left[ \mathrm{O}^{\mathrm{EST2}}_{A|Q} \in \sigma_{(i,b^{\bar{i}}, s^i , s^{\bar{i}})} \right]  \leq 4c\delta \;. \qedhere
	\]
\end{proof}

Lemma \ref{lem:strong_condition} tells us that if $\mathrm{P}_{\vec{A}|\vec{Q}}$ is a permutation-invariant non-signalling strategy then the probability that $\mathrm{O}^{\mathrm{EST2}}_{A|Q}$ is in a given set $\sigma_{(i,b^{\bar{i}}, s^i , s^{\bar{i}})}$ is exponentially small. In the next lemma we use this property to get a bound on the winning probability of $\mathrm{O}^{\mathrm{EST2}}_{A|Q}$ in the game $G$. 

\begin{lemma}\label{lem:distance_transformation}
	Let $\kappa =  \sum_{j=1}^{d} |y^{\star}_j|$ where $d$ is the number of signalling tests and $y^{\star}$ is an optimal solution of the dual program \eqref{eq:dual}. 
	Let $\mathrm{O}^{\mathrm{EST2}}_{A|Q}$ be a strategy such that for all $(i,b^{\bar{i}}, s^i , s^{\bar{i}})$, $\mathrm{O}^{\mathrm{EST2}}_{A|Q} \notin \sigma_{(i,b^{\bar{i}}, s^i , s^{\bar{i}})}$. Then $w(\mathrm{O}^{\mathrm{EST2}}_{A|Q})\leq 1-\alpha + \left(\zeta+2\epsilon\right)\kappa$.
\end{lemma}

\begin{proof}

	According to Lemma \ref{lem:small_sigma_signalling},  if $\mathrm{O}^{\mathrm{EST2}}_{A|Q} \notin \sigma_{(i,b^{\bar{i}}, s^i , s^{\bar{i}})}$ for every $\sigma_{(i,b^{\bar{i}}, s^i , s^{\bar{i}})}$ then 
	\begin{equation}\label{eq:perturbed_signalling}
		\mathrm{Sig}_{(i,a^{\bar{i}},q^i,q^{\bar{i}})}(\mathrm{O}^{\mathrm{EST2}})<\zeta+2\epsilon
	\end{equation}
	for every $i$ and every $b^{\bar{i}},s^i,s^{\bar{i}}$. That is, $\mathrm{O}^{\mathrm{EST2}}_{A|Q}$  is not ``too signalling'' in any direction. This can be used to bound the winning probability of $\mathrm{O}^{\mathrm{EST2}}_{A|Q}$ in the game $G$.
	
	The following linear program describes the optimal winning probability of a strategy $\mathrm{O}_{A|Q}$ which fulfils Equation~\eqref{eq:perturbed_signalling}: 
	\begin{equation} \label{eq:perturbation}
	\begin{aligned}
		\max \quad &\sum_{q,a} Q(q) R(q,a) \mathrm{O}(a|q) \\
		\text{s.t.} \quad&   Q(q^i , q^{\bar{i}})\left[ \mathrm{O}(\circ,a^{\bar{i}}|q^i , q^{\bar{i}}) - \sum_{r^{i}} Q(r^{i}|q^{\bar{i}})\mathrm{O}(\circ,a^{\bar{i}} | r^{i},q^{\bar{i}})\right] \leq \zeta+2\epsilon  &\forall i, q^i,q^{\bar{i}},a^{\bar{i}}  \\
		&\sum_{a} \mathrm{O}(a|q) = 1  &\forall q  \\
		& \mathrm{O}(a|q) \geq 0 &\forall a,q 
	\end{aligned}
	\end{equation}

	As $\mathrm{O}^{\mathrm{EST2}}_{A|Q}$ is a strategy we have $\sum_{a} \mathrm{O}^{\mathrm{EST2}}(a|q) = 1$ for all $q$. Hence $\mathrm{O}^{\mathrm{EST2}}_{A|Q}$ satisfies all the constraints of the above program and therefore its winning probability in $G$ is bounded by the optimal value of the program.
	Program \eqref{eq:perturbation} can be seen as a perturbation of the linear program~\eqref{eq:linear_program}, we can therefore bound its optimal value by using known tools for sensitivity analysis of linear programs, stated in Lemmas~\ref{lem:sensitivity1} and \ref{lem:sensitivity2}. 
	
	Denote by $y^{\star}$ an optimal solution of the dual program\footnote{We are only interested in the value of $y^{\star}$ as $z^{\star}$ will not affect the bound.} \eqref{eq:dual} and let $\kappa =  \sum_{j=1}^{d} |y^{\star}_j|$ where $d$ is the number of signalling tests. That is, $\kappa$ is the sum of all the dual variables which are associated to the non-signalling constraints. 
	
	According to Lemma \ref{lem:sensitivity1} the perturbed winning probability is then bounded by 
	\[
		w_e \leq 1-\alpha + \left(\zeta+2\epsilon\right)\kappa .\qedhere
	\]
\end{proof}

To get $\kappa$ in the above lemma, one can use any of the following:  
\begin{enumerate}
	\item Given a description of a game one can easily get $\kappa$ by solving the dual program~\eqref{eq:dual}\footnote{Solving the linear program is anyhow usually necessary for knowing the optimal non-signalling value $1-\alpha$.}. 
	\item If the game involves only 2 players, then following \cite{ito2010polynomial} one can get $\kappa \leq d$ where $d$ is the number of different signalling tests ($d<m|\mathcal{Q}||\mathcal{A}|$).
	\item Otherwise, the general bound of Lemma \ref{lem:sensitivity2} can be used. In our case the bound reads $\kappa \leq  |\mathcal{A}|^2|\mathcal{Q}|\Delta$, where $\Delta$ depends only on the game\footnote{A similar bound was also used in \cite{buhrman2013parallel}.}. 
\end{enumerate}	

Finally we are ready to prove the last lemma:
 \begin{lemma}[Main lemma]\label{lem:threshold}
	Let $w(G) = 1-\alpha$ be the optimal winning probability of a non-signalling strategy in $G$. Let $0<\beta\leq\alpha$ be some constant and $n$ the number of repetitions such that Equation \eqref{eq:rep_num_condition} is satisfied. Then for any non-signalling strategy $\mathrm{P}_{\vec{A}|\vec{Q}}$ of the repeated game, 
	\[
		\mathrm{Pr}_{\vec{a},\vec{q}\sim\mathrm{P}_{\vec{A}\vec{Q}}} \left[w(\mathrm{O}^{\mathrm{EST2}}_{A|Q}) > 1 - \alpha + \beta \right] \leq 5cd\delta \;.
	\]
\end{lemma}

\begin{proof}
	Let $\zeta,\epsilon>0$ be such that $\zeta+2\epsilon \leq \frac{\beta}{\kappa}$, $\epsilon \leq \min_q Q(q)$ and $7\epsilon \leq \zeta \leq 1$.
	
	If all tuples of questions $s$ appear at least once in the game data then according to the definition of  $\mathrm{O}^{\mathrm{EST2}}$ we have $\sum_{b} \mathrm{O}^{\mathrm{EST2}}(b|s) = 1$ for all $s$. We can therefore apply Lemma \ref{lem:distance_transformation} in combination with Lemma \ref{lem:strong_condition} and get
	\begin{align*}
		\mathrm{Pr}_{\vec{a},\vec{q} \sim \mathrm{P}_{\vec{A}\vec{Q}}} \left[ w(\mathrm{O}^{\mathrm{EST2}}_{A|Q}) > 1-\alpha + \beta \big| \sum_{b} \mathrm{O}^{\mathrm{EST2}}(b|s) = 1 \forall s \right] &\leq \mathrm{Pr}_{\vec{a},\vec{q} \sim \mathrm{P}_{\vec{A}\vec{Q}}} \left[ \exists \sigma_{(i,b^{\bar{i}}, s^i , s^{\bar{i}})} \text{ s.t. } \mathrm{O}^{\mathrm{EST2}}_{A|Q} \in \sigma_{(i,b^{\bar{i}}, s^i , s^{\bar{i}})}  \right]  \\
		&\leq  d\cdot 4c\delta \;. 
	\end{align*}
	
	The probability that some tuple of questions does not appear in the game data is upper bounded by 
	\[
		|\mathcal{Q}|\left(1-\min_s Q(s)\right)^{n/2} \leq |\mathcal{Q}|e^{-\min_s Q(s)n/2} \leq |\mathcal{Q}|e^{-\epsilon n/2} \leq d\delta
	\]
	and therefore all together we have
	\[
		\mathrm{Pr}_{\vec{a},\vec{q} \sim \mathrm{P}_{\vec{A}\vec{Q}}} \left[ w(\mathrm{O}^{\mathrm{EST2}}_{A|Q}) > 1-\alpha + \beta \right] \leq 5cd\delta\;. \qedhere 
	\]
\end{proof}

Our threshold theorem, Theorem \ref{thm:final_threshold_theorem}, follows from Lemma \ref{lem:threshold}:
\begin{proof}[Proof of Theorem \ref{thm:final_threshold_theorem}]
	Let $f^g,f^t$ and $f$ denote the winning frequency in the game data, test data and the entire data respectively (i.e., the fraction of coordinates in which the players win the game). Form Lemma  \ref{lem:threshold} we know that $\mathrm{Pr}_{\vec{a},\vec{q}\sim\mathrm{P}_{\vec{A}\vec{Q}}} \left[ f^g > 1 - \alpha + \beta \right] \leq 5cd\delta$, as the winning frequency in the game questions is exactly $w(\mathrm{O}^{\mathrm{EST2}}_{A|Q})$. As the game data and test data are symmetric (i.e., there is no difference between them except for the name we gave them), the same result also holds for $f^t$, $\mathrm{Pr}_{\vec{a},\vec{q}\sim\mathrm{P}_{\vec{A}\vec{Q}}} \left[ f^t > 1 - \alpha + \beta \right] \leq 5cd\delta$. 
	
	Finally, as the winning frequency in the entire data is given by $f=\frac{1}{2}\left( f^t + f^g \right)$ we have 
	\begin{equation}\label{eq:complete_statement}
		\mathrm{Pr}_{\vec{a},\vec{q}\sim\mathrm{P}_{\vec{A}\vec{Q}}} \left[ f > 1 - \alpha + \beta \right] \leq \mathrm{Pr}_{\vec{a},\vec{q}\sim\mathrm{P}_{\vec{A}\vec{Q}}} \left[ \left(f^t > 1 - \alpha + \beta\right) \lor \left( f^g > 1 - \alpha + \beta \right) \right]\leq  10cd\delta \;. \qedhere
	\end{equation}
	
\end{proof}

The relations between all the constants and parameters of the theorem and the proofs are listed in Table \ref{tb:parameters_table}. Note that for any game and choice of parameters the bound $10cd \delta$ is exponentially decreasing with the number of repetitions~$n$.


\begin{table}
	\begin{center}
    	\begin{tabular}{| l | l | l | l |}
   	 \hline
   		Symbol & Meaning & First appears & Fulfils  \\ \hhline{|=|=|=|=|}
		$m$ (g) & \# of players && \\ \hline   	
		$1-\alpha$ (g)& optimal NS winning probability in $G$ && \\ \hline
		$\kappa$ (g)& bound on an optimal dual solution $y^{\star}$ && $\sum_{j=1}^{d} |y^{\star}_j|$ \\ \hline
		$W_{\mathrm{ns}}$ (g)& optimal NS winning probability in guessing game&& $\max_{q,i} Q(q^i|q^{\bar{i}})$\\ \hhline{|=|=|=|=|}	
		$n$  (t)& \# of repetitions & & \\ \hline
		$\beta$ (t)& deviation in the threshold theorem &&  \\ \hhline{|=|=|=|=|}
    		$\epsilon$  & confidence interval of the test& Equation \eqref{eq:test_definition}& $\epsilon\leq\min_q Q(q)$ \\ \hline 
		$\zeta$  &signalling threshold of the test&Equation \eqref{eq:test_definition}& $7\epsilon\leq\zeta\leq 1$ ; $\zeta + 2\epsilon \leq \frac{\beta}{\kappa}$  \\ [2.5pt]\hline
		$\delta$ & confidence level of the test& Lemma \ref{lem:reliable_test}&$\delta=(n/2+1)^{|\mathcal{A}|\cdot|\mathcal{Q}|-1}e^{-n\epsilon^2/4}$ \\ \hline 
		$\nu$  &signalling threshold& Lemma \ref{lem:Sigma_signalling}&$\frac{2c\delta}{1-2c\delta} W_{\mathrm{ns}} <\nu< \zeta - 6\epsilon$   \\ \hline
		$c$ &de Finetti constant& Lemma \ref{lem:deFinetti_reduction}&$c=(n+1)^{|\mathcal{Q}|(|\mathcal{A}|-1)}$ \\ \hline
		$d$ &\# of different signalling tests& Lemma \ref{lem:distance_transformation}&$d<m|\mathcal{Q}||\mathcal{A}|$ \\ \hline
   	 \end{tabular}
	\end{center}\caption{Constants, parameters and their relations. (g) next to the symbol denotes that this is a constant which depends on the considered game and (t) denotes a parameter of the threshold theorem. All other constants should be chosen such that all the requirements in the last column of the table are fulfilled.}\label{tb:parameters_table}
\end{table}

To get a better feeling of the result, without trying to optimise it, one can make the following choices. Let $\epsilon = \frac{\beta}{10\kappa}$, $\zeta=8\epsilon$ and $\nu=\epsilon$ (assuming $\min_q Q(q) > \frac{\beta}{10\kappa}$). 
Using these choices, our proof holds for $n$ and $\beta$ such that 
\begin{equation}\label{eq:choice_of_number_condition}
	\frac{n}{\ln(n)}>20|\mathcal{Q}||\mathcal{A}|\frac{\ln(20\kappa/\beta)}{(\beta/10\kappa)^2}
\end{equation}
with the following constants in Theorem \ref{thm:final_threshold_theorem}:
\begin{equation}\label{eq:choice_of_constants}
	\mathcal{C}_1(G,n) = 10 m |\mathcal{Q}||\mathcal{A}| \left(n+1\right)^{2(|\mathcal{Q}||\mathcal{A}|-1 )} \;, \quad \mathcal{C}_2(G) = (20\kappa)^{-2} \;.
\end{equation}
The theorem then reads
\begin{equation}\label{eq:choice_of_param}
	\mathrm{Pr}_{\vec{a},\vec{q}\sim\mathrm{P}_{\vec{A}\vec{Q}}} \left[f > 1 - \alpha + \beta \right] \leq 10  m |\mathcal{Q}||\mathcal{A}| \left(n+1\right)^{2(|\mathcal{Q}||\mathcal{A}|-1)} e^{-\frac{n}{4}\left(\frac{\beta}{10 \kappa}\right)^2} \;.
\end{equation}
A different choice of parameters can improve the dependency of the constants on the game $G$.

\section{Conclusions and open questions}\label{sec:conclusions}

\subsection{Current work and possible extensions}

In this work a threshold theorem for multiplayer non-signalling games was proven. The threshold theorem given in Theorem \ref{thm:final_threshold_theorem} is applicable to any multiplayer complete-support game and for every two-player game (not necessarily with complete-support, as proven in Appendix \ref{app:two_player_extension}). Hence, all cases for which parallel repetition was already known prior to our work \cite{holenstein2007parallel,buhrman2013parallel} are covered by our proof. For multiplayer-games with incomplete-support we considered a small modification of the parallel repetition procedure which results in Theorem \ref{thm:mod_threshold_theorem}. We believe a similar modification can be considered to extend the result of \cite{buhrman2013parallel}.

In both theorems it might be possible to improve the dependency of the result on the parameters of the considered game, i.e., improve the constants $\mathcal{C}_1(G,n)$ and $\mathcal{C}_2(G)$. The polynomial dependency of $\mathcal{C}_1(G,n)$ on the number of repetitions, on the other hand, is inherent to the use of the de Finetti theorem. Moreover, further investigation of the dual program \eqref{eq:dual} could lead to an explicit bound on $\mathcal{C}_2(G)$. This could then be used to extend Theorem \ref{thm:final_threshold_theorem} to games with incomplete-support, as done for two-player games.

The most important contribution of this work is a new proof technique for parallel repetition theorems, based on ideas of de Finetti theorems and tomography. de Finetti theorems seem like a natural tool for parallel repetition theorems, yet, this is the first time that such a result is proven using a de Finetti theorem. 

Apart from allowing a different point of view on parallel repetition questions, and the study of correlations in general, the new proof technique has several advantages over the previous proofs. 

For instance, note that in the standard proofs of parallel repetition theorems, i.e., proofs following the approach of~\cite{raz1998parallel} such as \cite{holenstein2007parallel,rao2011parallel,buhrman2013parallel}, most of the difficulties arise due to the effect of conditioning on the event of winning some of the game repetitions. As this event is one that depends on the structure of the game and we have no control over it, it can introduce arbitrary correlations between the questions used in different repetitions of the game, a major source of difficulty for the remainder of the argument. In our proof we also need to analyse the effect of conditioning on a certain event, the event of the non-signalling test accepting, and this is done in Lemma \ref{lem:weak_condition}, the weak lemma. However, the key advantage of our approach is that the test has a very specific structure, and in particular conditioning on the test passing can be done locally by the players in a way that respects the non-signalling constraints. As a result it is almost trivial to deal with the conditioning in the remainder of the proof. This shift from conditioning on an uncontrolled event, success in the game, to a highly controlled one, a non-signalling test that we design ourselves, is a key simplification that we expect to play an important role in any extension of our method to other scenario such as classical or quantum strategies. 
More specifically, by finding appropriate ``non-classicality'' and ``non-quantumness'' measures which can replace our signalling measure in Definition \ref{def:signalling} one may be able to adapt the proof to the multiplayer classical and quantum cases as well.  The results of Sections \ref{sec:signalling} and \ref{sec:deFinetti} should follow easily for most ``non-classicality'' and ``non-quantumness'' measures of one-game strategies. The main difficulty, however, is finding a measure for which Lemma \ref{lem:weak_condition} can be proven.

\subsection{What parallel repetition tells us about de Finetti theorems}\label{sec:what_we_learn}

In the light of the de Finetti reduction stated in Lemma \ref{lem:deFinetti_reduction}, it is tempting to try and prove a parallel repetition theorem by claiming that for every permutation-invariant strategy $\mathrm{P}_{\vec{A}|\vec{Q}}$, 
\begin{equation}\label{eq:hypo_deFinetti_arguement}
	w\left(\mathrm{P}_{\vec{A}|\vec{Q}}\right) \leq c \cdot w\left(\tau_{\vec{A}|\vec{Q}}\right)\;.
\end{equation}
This claim is correct but, unfortunately, not very useful as $\tau_{\vec{A}|\vec{Q}}$ itself is signalling according to the explicit construction given in \cite{arnon2013finetti}, hence, no non-trivial bound holds on $w\big(\tau_{\vec{A}|\vec{Q}}\big)$.

One might hope that this is just a technical problem; perhaps a different de Finetti reduction can be proven, where both $\mathrm{P}_{\vec{A}|\vec{Q}}$ and $\tau_{\vec{A}|\vec{Q}}$ can be taken to be non-signalling (or analogously, quantum or classical). Such a de Finetti reduction, if it existed, would have implied a strong parallel repetition theorem (up to the polynomial factor $c$) for any game right away using Equation~\eqref{eq:hypo_deFinetti_arguement}. This however will stand in contradiction to known impossibility results, such as the result of \cite{kempe2010no}.

We therefore learn an interesting fact about de Finetti reductions by considering parallel repetition theorems: in order to prove a general de Finetti reduction as in Lemma \ref{lem:deFinetti_reduction}, the de Finetti strategy must have some signalling parts. Fortunately, as shown by our result, this does not render a proof for the non-signalling case impossible. 

\begin{acknowledgments}
The authors would like to thank Anthony Leverrier, Yeong-Cherng Liang, and David Sutter for helpful discussions, and Christian Schaffner for pointing out a mistake in a preliminary version of this work.  RAF and RR acknowledge support from the European Research Council (ERC) via grant No.~258932, the European Commission STReP project ``RAQUEL'',  the Swiss National Science Foundation (SNSF) via the National Centre of Competence in Research ``QSIT'' and the CHIST- ERA project ``DIQIP''. TV acknowledges support from the Simons Institute in Berkeley, the Ministry of Education, Singapore under the Tier 3 grant~MOE2012-T3-1-009 and the Perimeter Institute in Waterloo, Canada.
\end{acknowledgments}

\bibliography{refs.bib}

\appendix

\section{Extending the result to general games}\label{app:extension}

Before we show how to extend the threshold theorem to games with incomplete-support, let us explain why the proof given for Theorem \ref{thm:final_threshold_theorem} holds only for complete-support games.

As mentioned in the main text, the reason lies in the linear program \eqref{eq:linear_non_relaxed}, and more specifically, in the non-signalling constraints given in Equation \eqref{eq:linear1_ns}. Indeed, if for some $q$ we have $Q(q) = 0$ then the relevant constraint in Equation~\eqref{eq:linear1_ns} is vacuous. It is therefore clear that in this case the constraints given in Equation \eqref{eq:linear1_ns} are in fact relaxations of the standard non-signalling constraints given in Equation~\eqref{eq:normal_ns_def}.

For some games, this relaxation of the non-signalling constraints is strict. For example\footnote{This example was communicated to us by Christian Schaffner.}, consider a game of 3 players where the questions are uniformly distributed over $\mathcal{Q} = \{ (0,0,1), (0,1,0), (1,0,0)\}$ and the winning condition is given by the following predicate: 
\[
	R(q,a) = 
	\begin{cases}
		1 & \text{if } q=(0,0,1) \text{ and } a^1= a^2 \\  
		1 & \text{if } q=(0,1,0) \text{ and } a^1 = a^3 \\  		
		1 & \text{if } q=(1,0,0) \text{ and } a^2\neq a^3 \\
		0& \text{otherwise}
	\end{cases} \\
\]
The optimal non-signalling winning probability in this game is $\frac{2}{3}$ (as can be shown by solving a linear program). However, in the linear program \eqref{eq:linear_non_relaxed} there are no non-trivial constraints (i.e., all the constraints in Equation~\eqref{eq:linear1_ns} are of the form $0=0$). Hence, the optimal solution of the program~\eqref{eq:linear_non_relaxed} is 1, which is strictly larger than $\frac{2}{3}$. Thus even though the non-signalling conditions are enforced over all ``relevant'' questions, this does not suffice to guarantee that 
there exists a strategy achieving the resulting optimum success probability $1$ and that can be extended to a non-signalling strategy defined on all questions. 

For games with incomplete-support in which the optimal value of program \eqref{eq:linear_non_relaxed} is not trivial (i.e., it is smaller than~1), it follows that our proof can be applied as is to derive a non-trivial threshold theorem. Irrespectively of whether this is the case or not one might also elect to work with the weaker definition of non-signalling strategies that is implied by the constraints in~\eqref{eq:linear1_ns}, where the behaviour of the strategy is not required to be well-defined for questions which do not appear in the game. In this case the linear program \eqref{eq:linear_non_relaxed} exactly describes the optimal winning probability of such strategies and Theorem \ref{thm:final_threshold_theorem} holds without any modification.

In other cases, on the other hand, we have to slightly modify the linear program in order to derive a correct threshold theorem. In the following sections we show how to do this. 

\subsection{Two-player games}\label{app:two_player_extension}

For two-player games we consider the following modification of the linear program \eqref{eq:linear_non_relaxed}. 
\begin{equation} \label{eq:mod_linear_non_relaxed}
\begin{aligned}
	\max \quad &\sum_{q,a} Q(q) R(q,a) \mathrm{O}(a|q)  \\
	\text{s.t.} \quad&   Q(q^i , q^{\bar{i}})\left[ \mathrm{O}(\circ,a^{\bar{i}}|q^i , q^{\bar{i}}) - \sum_{r^{i}} Q(r^{i}|q^{\bar{i}})\mathrm{O}(\circ,a^{\bar{i}} | r^{i},q^{\bar{i}})\right] = 0  &\forall i, a^{\bar{i}}, \forall q^i,q^{\bar{i}} \text{ s.t. } Q(q)\neq 0  \\
	&   \eta \left[ \mathrm{O}(\circ,a^{\bar{i}}|q^i , q^{\bar{i}}) - \sum_{r^{i}} Q(r^{i}|q^{\bar{i}})\mathrm{O}(\circ,a^{\bar{i}} | r^{i},q^{\bar{i}})\right] = 0  &\forall i, a^{\bar{i}}, \forall q^i,q^{\bar{i}} \text{ s.t. } Q(q)= 0 \\
	&\sum_{a} \mathrm{O}(a|q) = 1  &\forall q \\
	& \mathrm{O}(a|q) \geq 0 &\forall a,q 
\end{aligned} 
\end{equation}
where $\eta>0$ is some small constant that will be chosen later. 

Following the analysis proposed in \cite{ito2010polynomial} (Section 4 therein), one can show that the program \eqref{eq:mod_linear_non_relaxed} can be relaxed to the following equivalent program:

\begin{subequations} \label{eq:mod_linear_relaxed}
\begin{align}
	\max \quad &\sum_{q,a} Q(q) R(q,a) \mathrm{O}(a|q) \nonumber \\
	\text{s.t.} \quad&   Q(q^i , q^{\bar{i}})\left[ \mathrm{O}(\circ,a^{\bar{i}}|q^i , q^{\bar{i}}) - \sum_{r^{i}} Q(r^{i}|q^{\bar{i}})\mathrm{O}(\circ,a^{\bar{i}} | r^{i},q^{\bar{i}})\right] \leq 0  &\forall i, a^{\bar{i}}, \forall q^i,q^{\bar{i}} \text{ s.t. } Q(q)\neq 0 \label{eq:linear1_ns_app} \\
	&   \eta \left[ \mathrm{O}(\circ,a^{\bar{i}}|q^i , q^{\bar{i}}) - \sum_{r^{i}} Q(r^{i}|q^{\bar{i}})\mathrm{O}(\circ,a^{\bar{i}} | r^{i},q^{\bar{i}})\right] \leq 0  &\forall i, a^{\bar{i}}, \forall q^i,q^{\bar{i}} \text{ s.t. } Q(q)= 0\label{eq:linear1_ns_mod_app} \\
	&\sum_{a} \mathrm{O}(a|q) \leq 1  &\forall q \nonumber \\
	& \mathrm{O}(a|q) \geq 0 &\forall a,q \nonumber
\end{align} 
\end{subequations}

Moreover, following \cite{ito2010polynomial} it can also be shown that the dual variables $y^\star$ which are associated with the primal constraints of Equations \eqref{eq:linear1_ns_app} and \eqref{eq:linear1_ns_mod_app} are all upper bounded by 1, independently of the value of $\eta$. This implies that $\kappa =  \sum_{j=1}^{d} |y^{\star}_j| \leq d$ is also independent of $\eta$ (where $d$ is now the total number of constraints in Equations~\eqref{eq:linear1_ns_app} and~\eqref{eq:linear1_ns_mod_app} together).

When applying our proof using the linear program \eqref{eq:mod_linear_relaxed} we get the following perturbed linear program in Lemma~\ref{lem:distance_transformation} (instead of the one given in Equation \eqref{eq:perturbation}):

\begin{subequations} \label{eq:mod_linear_relaxed_perturbed}
\begin{align}
	\max \quad &\sum_{q,a} Q(q) R(q,a) \mathrm{O}(a|q) \nonumber \\
	\text{s.t.} \quad&   Q(q^i , q^{\bar{i}})\left[ \mathrm{O}(\circ,a^{\bar{i}}|q^i , q^{\bar{i}}) - \sum_{r^{i}} Q(r^{i}|q^{\bar{i}})\mathrm{O}(\circ,a^{\bar{i}} | r^{i},q^{\bar{i}})\right] \leq \zeta + 2\epsilon  &\forall i, a^{\bar{i}}, \forall q^i,q^{\bar{i}} \text{ s.t. } Q(q)\neq 0 \label{eq:linear1_ns_perturbed} \\
	&   \eta \left[ \mathrm{O}(\circ,a^{\bar{i}}|q^i , q^{\bar{i}}) - \sum_{r^{i}} Q(r^{i}|q^{\bar{i}})\mathrm{O}(\circ,a^{\bar{i}} | r^{i},q^{\bar{i}})\right] \leq \zeta + 2\epsilon  &\forall i, a^{\bar{i}}, \forall q^i,q^{\bar{i}} \text{ s.t. } Q(q)= 0\label{eq:linear1_ns_mod_perturbed} \\
	&\sum_{a} \mathrm{O}(a|q) \leq 1  &\forall q \label{eq:linear1_normal_perturbed} \\
	& \mathrm{O}(a|q) \geq 0 &\forall a,q \nonumber
\end{align} 
\end{subequations}

The estimated strategy $\mathrm{O}^{\mathrm{EST2}}_{A|Q}$ fulfils the constraints of Equation \eqref{eq:linear1_ns_perturbed} as in the proof in the main text. Moreover, it fulfils Equation \eqref{eq:linear1_normal_perturbed} by definition (see Section \ref{sec:estimated_strategies}). 
Therefore, in order to ensure that the winning probability of $\mathrm{O}^{\mathrm{EST2}}_{A|Q}$ is bounded by the optimal value of the program \eqref{eq:mod_linear_relaxed_perturbed} we only need to choose $\eta \leq \zeta + 2\epsilon$ such that the constraints of Equation \eqref{eq:linear1_ns_mod_perturbed} will hold as well. 

To see that this is possible, recall that the values of $\zeta$ and $\epsilon$ are chosen such that $\zeta + 2\epsilon \leq \frac{\beta}{\kappa}$. As both $\beta$ and $\kappa$ are independent of $\eta$ we can just choose $\eta \leq \zeta + 2\epsilon$. The rest of the proof then follows in the same way as in the main text and Theorem \ref{thm:final_threshold_theorem} is derived (without any dependence on $\eta$).

\subsection{General games}\label{app:general_extension}

As the technique of the previous section is relevant only for two-player games\footnote{To be more precise, it holds for any game where $\kappa$ can be bounded by a constant independent of the questions distribution~$Q$.}, the aim of this section is to explain how our proof can be adapted to derive a useful result for multiplayer games which do not have complete-support, as stated in Theorem \ref{thm:mod_threshold_theorem}. To do so we slightly modify the parallel repetition procedure. 

Instead of considering the usual parallel repetition, in which $n$ tuples of questions are chosen according to the game distribution $Q$, we change the distribution of questions in the repeated game by sometimes (with small positive probability) asking the players a tuple of questions $q$ for which $Q(q)=0$. We call such questions ``dummy questions''; for these questions any answer from the players is accepted. The remaining questions, for which $Q(q)>0$, are called the ``real questions''. We denote the modified repeated game by $\tilde{G}^n$.  

It is important to note that the standard definition of the non-signalling constraints implies that a non-signalling strategy should have a well-defined behaviour for all possible inputs. As the referee ignores the players' answers to the additional questions, the specific behaviour of the strategy on dummy questions is irrelevant. Therefore, if the optimal non-signalling winning probability in $G$ is~1, then the winning probability in both $G^n$ and $\tilde{G}^n$ is also 1: our modification does not harm the success probability of ``honest'' players.

To prove Theorem \ref{thm:mod_threshold_theorem} we proceed in two steps: we make a small change in the linear program \eqref{eq:linear_program} and then apply our proof using the modified program. 

\subsubsection{Changing the linear program}

As a first step we define $\tilde{Q}$ to be a complete-support version of $Q$ in the following way\footnote{In \cite{feige1994two,kempe2011parallel} a subset of indices in which dummy, or ``confusion'', questions are asked is chosen. We choose to make a small modification in the questions distribution instead, such that permutation invariance is maintained.}.

Let $\mathrm{I}(q)$ be the indicator function such that $\mathrm{I}(q)=1$ if $q$ is a dummy question, i.e., if $Q(q)=0$, and~1 otherwise. Denote by $\mathcal{D}$ the number of dummy questions $\mathcal{D} = |\{ q | \mathrm{I}(q)=1\}|$. 

Let $\eta>0$ be some small constant (which can be later chosen to optimise the bound obtained in the final result).  We define the following joint probability distribution of $q$ and $d\in\{0,1\}$:
\[
	\mathrm{P}_{\tilde{Q}D}(q,d) = 
	\begin{cases}
    		\frac{\eta}{\mathcal{D}} & \text{if } \mathrm{I}(q)=1 \text{ and } d=1 \\
    		Q(q)(1-\eta)& \text{if } \mathrm{I}(q)=0 \text{ and } d=0 \\
		0& \text{otherwise}
	\end{cases} \\
\]

Then $\tilde{Q}(q)= \sum_{d\in\{0,1\}}\mathrm{P}_{\tilde{Q}D}(q,d)$ and we have 
\[
	\mathrm{P}_{\tilde{Q}|D=0}(q) = \frac{\mathrm{P}_{\tilde{Q}D}(q,0)}{\sum_q \mathrm{P}_{\tilde{Q}D}(q,0)} = \frac{\mathrm{P}_{\tilde{Q}D}(q,0)}{1-\eta} = Q(q) \;.
\]
That is, when conditioning on the event of a question not being a dummy question we retrieve $Q$ from~$\tilde{Q}$. 

Next, we use $\tilde{Q}$ to write the non-signalling constraints (but keep $Q$ in the objective function):
\begin{equation} \label{eq:linear_prog}
\begin{aligned} 
	\max \quad &\sum_{q,a} Q(q) R(q,a) \mathrm{O}(a|q) \\
	\text{s.t.} \quad&   \tilde{Q}(q^i , q^{\bar{i}})\left[ \mathrm{O}(\circ,a^{\bar{i}}|q^i , q^{\bar{i}}) - \sum_{r^{i}} \tilde{Q}(r^{i}|q^{\bar{i}})\mathrm{O}(\circ,a^{\bar{i}} | r^{i},q^{\bar{i}})\right] \leq 0  &\forall i, q^i,q^{\bar{i}},a^{\bar{i}}  \\
	&\sum_{a} \mathrm{O}(a|q) = 1  &\forall q  \\
	& \mathrm{O}(a|q) \geq 0 &\forall a,q 
\end{aligned}
\end{equation}

This linear program replaces the program \eqref{eq:linear_program}. 
The distance measure in Definition \ref{def:trace_distance} and the signalling measure in Definition \ref{def:signalling} should now be defined with respect to $\tilde{Q}$ as well. 

\subsubsection{Deriving Theorem \ref{thm:mod_threshold_theorem}}

Following the proof of Theorem \ref{thm:final_threshold_theorem} with the above changes we get the following statement in the main Lemma, Lemma \ref{lem:threshold}:
\begin{equation}\label{eq:main}
	\mathrm{Pr}_{\vec{a},\vec{q}\sim\mathrm{P}_{\vec{A}\vec{\tilde{Q}}}} \left[w(\mathrm{O}^{\mathrm{EST2}}_{A|Q}) > 1 - \alpha + \beta \right] \leq 5cd\delta \;.
\end{equation}
where the data $\vec{a},\vec{q}$ is now distributed according to $\mathrm{P}_{\vec{A}\vec{\tilde{Q}}}= \tilde{Q}^{\otimes n}\times\mathrm{P}_{\vec{A}|\vec{Q}}$ and the parameter $\delta$ now depends on the change we did in the question distribution $\tilde{Q}$ (through $\kappa$ which depends on the solution of the dual program of program~\eqref{eq:linear_prog}, and thus has an implicit dependence on $\eta$). 

As the objective function of program \eqref{eq:linear_prog} is given using $Q$ and not $\tilde{Q}$,  $w(\mathrm{O}^{\mathrm{EST2}}_{A|Q})$ in Equation~\eqref{eq:main} is the winning probability with respect to the original question distribution $Q$. It is therefore equal to the winning frequency in the real questions (i.e. it does not take the indices where dummy questions were asked into account). Hence, it leads to the desired statement:

\[
	\mathrm{Pr}_{\vec{a},\vec{q}\sim\mathrm{P}_{\vec{A}\vec{\tilde{Q}}}} \left[ f > 1 - \alpha + \beta \right] \leq  10cd\delta \;,
\]
where $f$ is the winning frequency in the real questions. This proves Theorem~\ref{thm:mod_threshold_theorem}.  

The parameter $\eta$ can be optimised in different ways, depending on the application. If one is interested in the bound itself and is not concerned by the modification of the repeated game the precise value of $\eta$ should be chosen in order to optimise the constants $\mathcal{C}_1(G,n)$  and $\mathcal{C}_2(G)$ appearing in the bound. 
Alternatively, if one does not wish to change the game by too much, small values for $\eta$ will ensure that $\tilde{G}^n$ is relatively close to $G^n$ (due to the definition of $\tilde{Q}$ above). A smaller $\eta$ will lead to a smaller fraction of dummy questions, but could result in worse constants $\mathcal{C}_2(G)$. 

\section{Proofs of Section \ref{sec:signalling}}\label{app:signalling_proofs}
In this section we present all the proofs which are relevant to the signalling measures and signalling tests. 

The first proof is a proof of Lemma \ref{lem:continuity} which shows that the signalling measure given in Definition~\ref{def:signalling} is continuous. We repeat Lemma \ref{lem:continuity} here:

\begin{customlemma}{\ref{lem:continuity}}
	Let $\mathrm{O}_1$ and $\mathrm{O}_2$ be two one-game strategies such that $\big|\mathrm{O}_1-\mathrm{O}_2\big|_1 \leq \epsilon$. Then 
	\[
		\forall i, a^{\bar{i}}, q^i , q^{\bar{i}} \quad  \big|\mathrm{Sig}_{(i, a^{\bar{i}}, q^i , q^{\bar{i}} )}\left( \mathrm{O}_1 \right) - \mathrm{Sig}_{(i, a^{\bar{i}}, q^i , q^{\bar{i}} )}\left( \mathrm{O}_2 \right)\big|\leq 2\epsilon \;.
	\]
\end{customlemma}

\begin{proof}
	We prove a stronger result from which the lemma follows. We prove 
	\[
		\forall i \quad  \sum_{a^{\bar{i}},q} \big|\mathrm{Sig}_{(i, a^{\bar{i}}, q^i , q^{\bar{i}} )}\left( \mathrm{O}_1 \right) - \mathrm{Sig}_{(i, a^{\bar{i}}, q^i , q^{\bar{i}} )}\left( \mathrm{O}_2 \right)\big|\leq 2\epsilon \;.
	\]

	To do so first note the following,
	\begin{align*}
		\big|\mathrm{O}_1-\mathrm{O}_2\big|_1 &= \mathbb{E}_{q} \sum_{a} \big| \mathrm{O}_1(a|q) - \mathrm{O}_2(a|q) \big| \\
		& \geq \mathbb{E}_{q} \sum_{a^{\bar{i}}} \Big| \sum_{a^i} \left( \mathrm{O}_1(a^i,a^{\bar{i}}|q) - \mathrm{O}_2(a^i,a^{\bar{i}}|q) \right) \Big| \\
		& = \mathbb{E}_{q} \sum_{a^{\bar{i}}} \big| \mathrm{O}_1(\circ,a^{\bar{i}}|q) - \mathrm{O}_2(\circ,a^{\bar{i}}|q) \big| \\
		& =  \sum_{a^{\bar{i}},q} Q(q) \big| \mathrm{O}_1(\circ,a^{\bar{i}}|q) - \mathrm{O}_2(\circ,a^{\bar{i}}|q) \big| \;,
	\end{align*}
	therefore if $\big|\mathrm{O}_1-\mathrm{O}_2\big|_1 \leq \epsilon$ then 
	\begin{equation}\label{eq:continuity_proof}
		\sum_{a^{\bar{i}},q} Q(q) \big| \mathrm{O}_1(\circ,a^{\bar{i}}|q) - \mathrm{O}_2(\circ,a^{\bar{i}}|q) \big|  \leq \epsilon \;.
	\end{equation}

	Next, using Equation \eqref{eq:sig1}
	\begin{align*}
		& \sum_{a^{\bar{i}},q}  \big|\mathrm{Sig}_{(i, a^{\bar{i}}, q^i , q^{\bar{i}} )}\left( \mathrm{O}_1\right) - \mathrm{Sig}_{(i, a^{\bar{i}}, q^i , q^{\bar{i}} )}\left( \mathrm{O}_2 \right)\big|  \\ 
		& = \sum_{a^{\bar{i}},q} Q(q^i , q^{\bar{i}}) \Big|  \mathrm{O}_1(\circ,a^{\bar{i}}|q^i , q^{\bar{i}}) - \sum_{r^{i}} Q(r^{i}|q^{\bar{i}})\mathrm{O}_1(\circ,a^{\bar{i}} | r^{i},q^{\bar{i}}) - 
		 \mathrm{O}_2(\circ,a^{\bar{i}}|q^i , q^{\bar{i}}) + \sum_{r^{i}} Q(r^{i}|q^{\bar{i}})\mathrm{O}_2(\circ,a^{\bar{i}} | r^{i},q^{\bar{i}}) \Big| \\
		& =\sum_{a^{\bar{i}},q}  Q(q^i , q^{\bar{i}}) \Big|  \mathrm{O}_1(\circ,a^{\bar{i}}|q^i , q^{\bar{i}}) - \mathrm{O}_2(\circ,a^{\bar{i}}|q^i , q^{\bar{i}}) + \sum_{r^{i}} Q(r^{i}|q^{\bar{i}}) \left( \mathrm{O}_2(\circ,a^{\bar{i}} | r^{i},q^{\bar{i}}) - \mathrm{O}_1(\circ,a^{\bar{i}} | r^{i},q^{\bar{i}})\right) \Big| \\
		& \leq \sum_{a^{\bar{i}},q} Q(q^i , q^{\bar{i}}) \Big|  \mathrm{O}_1(\circ,a^{\bar{i}}|q^i , q^{\bar{i}}) - \mathrm{O}_2(\circ,a^{\bar{i}}|q^i , q^{\bar{i}})\Big| + \sum_{a^{\bar{i}},q} Q(q^i , q^{\bar{i}}) \Big|\sum_{r^{i}} Q(r^{i}|q^{\bar{i}}) \left( \mathrm{O}_2(\circ,a^{\bar{i}} | r^{i},q^{\bar{i}}) - \mathrm{O}_1(\circ,a^{\bar{i}} | r^{i},q^{\bar{i}})\right) \Big| \\
		& \leq \sum_{a^{\bar{i}},q} Q(q^i , q^{\bar{i}}) \Big|  \mathrm{O}_1(\circ,a^{\bar{i}}|q^i , q^{\bar{i}}) - \mathrm{O}_2(\circ,a^{\bar{i}}|q^i , q^{\bar{i}})\Big| +  \sum_{a^{\bar{i}},q} \sum_{r^{i}} Q(r^{i}|q^{\bar{i}}) Q(q^i , q^{\bar{i}})  \Big| \mathrm{O}_2(\circ,a^{\bar{i}} | r^{i},q^{\bar{i}}) - \mathrm{O}_1(\circ,a^{\bar{i}} | r^{i},q^{\bar{i}}) \Big| \\
		&= \sum_{a^{\bar{i}},q} Q(q^i , q^{\bar{i}}) \Big|  \mathrm{O}_1(\circ,a^{\bar{i}}|q^i , q^{\bar{i}}) - \mathrm{O}_2(\circ,a^{\bar{i}}|q^i , q^{\bar{i}})\Big| +  \sum_{a^{\bar{i}},q^{\bar{i}}} \sum_{r^{i}} Q(r^{i}|q^{\bar{i}}) Q(q^{\bar{i}})  \Big| \mathrm{O}_2(\circ,a^{\bar{i}} | r^{i},q^{\bar{i}}) - \mathrm{O}_1(\circ,a^{\bar{i}} | r^{i},q^{\bar{i}}) \Big| \\
		&= \sum_{a^{\bar{i}},q} Q(q^i , q^{\bar{i}}) \Big|  \mathrm{O}_1(\circ,a^{\bar{i}}|q^i , q^{\bar{i}}) - \mathrm{O}_2(\circ,a^{\bar{i}}|q^i , q^{\bar{i}})\Big| +  \sum_{a^{\bar{i}},q} Q(q^{i},q^{\bar{i}})  \Big| \mathrm{O}_2(\circ,a^{\bar{i}} | q^{i},q^{\bar{i}}) - \mathrm{O}_1(\circ,a^{\bar{i}} | q^{i},q^{\bar{i}}) \Big| \\
		& \leq 2 \epsilon 
	\end{align*} 
 	where the last inequality follows from Equation \eqref{eq:continuity_proof}. \qedhere
\end{proof}


Next we give the proof of Lemma \ref{lem:reliable_test}:

\begin{customlemma}{\ref{lem:reliable_test}}
	Assume the players share an i.i.d.\ strategy $\mathrm{O}_{A|Q}^{\otimes n}$ and let $\zeta,\epsilon>0$ be the the parameters defined  as in Equation~\eqref{eq:test_definition}. For every $(i,b^{\bar{i}}, s^i , s^{\bar{i}})$,
	\begin{enumerate}
		\item If $\mathrm{Sig}_{(i, b^{\bar{i}}, s^i , s^{\bar{i}} )}\left( \mathrm{O}\right)\geq \zeta$ then 
		\begin{equation}\label{eq:reliable_p1}
			\mathrm{Pr}_{\vec{a},\vec{q}\sim\mathrm{O}^{\otimes n}_{AQ}} \left[ \mathcal{T}_{(i,b^{\bar{i}}, s^i , s^{\bar{i}})}\left( \vec{q^t},\vec{a^t} \right) = 1 \right] > 1- \delta
		\end{equation}
		\item If $\mathrm{Sig}_{(i, b^{\bar{i}}, s^i , s^{\bar{i}} )}\left( \mathrm{O}\right)= 0$ then 
		\begin{equation}\label{eq:reliable_p2}
			\mathrm{Pr}_{\vec{a},\vec{q}\sim\mathrm{O}^{\otimes n}_{AQ}} \left[ \mathcal{T}_{(i,b^{\bar{i}}, s^i , s^{\bar{i}})}\left( \vec{q^t},\vec{a^t} \right) = 0 \right] > 1 - \delta
		\end{equation}
	\end{enumerate}
	where $\delta=\delta\left(\frac{n}{2},\epsilon\right)=\left(\frac{n}{2}+1\right)^{|\mathcal{A}|\cdot|\mathcal{Q}|-1}e^{-n\epsilon^2/4}$.
\end{customlemma}

\begin{proof}
	For the first part of the lemma assume that $\mathrm{Sig}_{(i, b^{\bar{i}}, s^i , s^{\bar{i}} )}\left( \mathrm{O}\right) \geq \zeta$. Then
	\begin{align*}
		\mathrm{Pr}_{\vec{a},\vec{q}\sim\mathrm{O}^{\otimes n}_{AQ}} \left[ \mathcal{T}_{(i,b^{\bar{i}}, s^i , s^{\bar{i}})}\left( \vec{q^t},\vec{a^t} \right) = 0 \right] &= 
		\mathrm{Pr}_{\vec{a},\vec{q}\sim\mathrm{O}^{\otimes n}_{AQ}} \left[  \mathrm{Sig}_{(i, b^{\bar{i}}, s^i , s^{\bar{i}} )}\left( \mathrm{O}^{\mathrm{EST1}} \right) < \zeta - 2\epsilon  \right] \\
		&\leq  \mathrm{Pr}_{\vec{a},\vec{q}\sim\mathrm{O}^{\otimes n}_{AQ}} \left[ |\mathrm{O}^{\mathrm{EST1}} -  \mathrm{O}|_1 > \epsilon \right] \\
		&\leq \delta
	\end{align*}
	where the first inequality is due to Lemma \ref{lem:continuity} and the second due to Lemma \ref{lem:chernoff}. This implies Equation~\eqref{eq:reliable_p1}. Equation \eqref{eq:reliable_p2} can be proven in an analogous way. 
\end{proof}


The last proof of this section is the proof of Lemma \ref{lem:Sigma_signalling}:

\begin{customlemma}{\ref{lem:Sigma_signalling}}
	Let $\nu > 0$ be any parameter such that $\nu < \zeta - 6\epsilon$. Then for every $(i, b^{\bar{i}}, s^i , s^{\bar{i}} )$,
	\[
	 	\forall \mathrm{O}\in\Sigma_{(i,b^{\bar{i}}, s^i , s^{\bar{i}})}, \quad \mathrm{Sig}_{(i, b^{\bar{i}}, s^i , s^{\bar{i}} )}\left( \mathrm{O}\right) > \nu \;.
	\]
\end{customlemma}
 
\begin{proof}
	Assume by contradiction that there exists $\mathrm{O}\in\Sigma_{(i,b^{\bar{i}}, s^i , s^{\bar{i}})}$ such that $\mathrm{Sig}_{(i, b^{\bar{i}}, s^i , s^{\bar{i}} )}\left( \mathrm{O}\right) \leq \nu$. Since $\mathrm{O}\in\Sigma_{(i,b^{\bar{i}}, s^i , s^{\bar{i}})}$ there exists $\bar{\mathrm{O}}$ such that $ |\mathrm{O} - \bar{\mathrm{O}}|_1 \leq \epsilon$ and 
	\begin{equation}\label{eq:nu_assumption}
		\mathrm{Pr}_{\vec{a},\vec{q}\sim\bar{\mathrm{O}}^{\otimes n}_{AQ}} \left[ \mathcal{T}_{(i,b^{\bar{i}}, s^i , s^{\bar{i}})}\left( \vec{q^t},\vec{a^t} \right) = 1 \right] > \delta\;.
	\end{equation}
	
	Using Lemma \ref{lem:continuity} we get $\mathrm{Sig}_{(i, b^{\bar{i}}, s^i , s^{\bar{i}} )}\left( \bar{\mathrm{O}}\right) \leq \nu + 2\epsilon$. 
	
	From Lemma \ref{lem:chernoff} we know that $ \mathrm{Pr}_{\vec{a},\vec{q}\sim\bar{\mathrm{O}}^{\otimes n}_{AQ}} \left[ |\bar{\mathrm{O}}^{\mathrm{EST1}} -  \bar{\mathrm{O}}|_1 > \epsilon \right] \leq \delta$ and therefore, using Lemma \ref{lem:continuity} again, 
	\[
		\mathrm{Pr}_{\vec{a},\vec{q}\sim\bar{\mathrm{O}}^{\otimes n}_{AQ}} \left[  \mathrm{Sig}_{(i, b^{\bar{i}}, s^i , s^{\bar{i}} )}\left( \bar{\mathrm{O}}^{\mathrm{EST1}} \right) > \nu + 4\epsilon  \right] \leq \delta \;.
	\]
	Since $\nu  < \zeta - 6\epsilon$ this implies 
	\[
		\mathrm{Pr}_{\vec{a},\vec{q}\sim\bar{\mathrm{O}}^{\otimes n}_{AQ}} \left[  \mathrm{Sig}_{(i, b^{\bar{i}}, s^i , s^{\bar{i}} )}\left( \bar{\mathrm{O}}^{\mathrm{EST1}} \right) > \zeta - 2\epsilon  \right] \leq \delta 
	\]
	and therefore, according to the definition of the test,
	\[
		\mathrm{Pr}_{\vec{a},\vec{q}\sim\bar{\mathrm{O}}^{\otimes n}_{AQ}} \left[ \mathcal{T}_{(i,b^{\bar{i}}, s^i , s^{\bar{i}})}\left( \vec{q^t},\vec{a^t} \right) = 1 \right] \leq \delta \;,
	\]
	which contradicts Equation \eqref{eq:nu_assumption}.
\end{proof}

\section{Proofs of Section \ref{sec:deFinetti}}\label{app:proofs_de_finetti}

In this section we prove the relevant properties of the de Finetti strategy. We prove Lemma \ref{lem:de_finetti_prop}:

\begin{customlemma}{\ref{lem:de_finetti_prop}}
	For a de Finetti strategy $\tau_{\vec{A}|\vec{Q}}$ and every $(i,b^{\bar{i}}, s^i , s^{\bar{i}})$
	\begin{enumerate}
		\item $\mathrm{Pr}_{\vec{a},\vec{q}\sim\tau_{\vec{A}\vec{Q}}} \left[\mathcal{T}_{(i,b^{\bar{i}}, s^i , s^{\bar{i}})} (\vec{q^t},\vec{a^t})=1 \land \mathrm{O}^{\mathrm{EST2}}_{A|Q} \notin \Sigma_{(i,b^{\bar{i}}, s^i , s^{\bar{i}})}  \right] \leq 2\delta $
		\item $\mathrm{Pr}_{\vec{a},\vec{q}\sim\tau_{\vec{A}\vec{Q}}} \left[ \mathcal{T}_{(i,b^{\bar{i}}, s^i , s^{\bar{i}})} (\vec{q^t},\vec{a^t})=0 \land \mathrm{O}^{\mathrm{EST2}}_{A|Q} \in \sigma_{(i,b^{\bar{i}}, s^i , s^{\bar{i}})}  \right] \leq 2\delta $
	\end{enumerate}
\end{customlemma}

\begin{proof}
	Since a de Finetti strategy is a convex combination of i.i.d.\ strategies, it is sufficient to prove this for i.i.d.\ strategies $\mathrm{O}_{A|Q}^{\otimes n}$ and the lemma will follow. We start by proving the first part of the lemma.
	
	If $\mathrm{Pr}_{\vec{a},\vec{q}\sim\mathrm{O}^{\otimes n}_{AQ}} \left[ \mathcal{T}_{(i,b^{\bar{i}}, s^i , s^{\bar{i}})} (\vec{q^t},\vec{a^t}) = 1 \right] \leq \delta$ then we are done. Consider therefore states $\mathrm{O}_{A|Q}$ such that 
	\[
		\mathrm{Pr}_{\vec{a},\vec{q}\sim\mathrm{O}^{\otimes n}_{AQ}} \left[ \mathcal{T}_{(i,b^{\bar{i}}, s^i , s^{\bar{i}})} (\vec{q^t},\vec{a^t}) = 1 \right] > \delta \;.
	\]
	For such states
	\[
		\mathrm{Pr}_{\vec{a},\vec{q}\sim\mathrm{O}^{\otimes n}_{AQ}} \left[ \mathrm{O}^{\mathrm{EST2}}_{A|Q} \notin \Sigma_{(i,b^{\bar{i}}, s^i , s^{\bar{i}})}  \right] \leq \mathrm{Pr}_{\vec{a},\vec{q}\sim\mathrm{O}^{\otimes n}_{AQ}}\left[ |\mathrm{O}^{\mathrm{EST2}}_{A|Q} -  \mathrm{O}_{A|Q}|_1 > \epsilon \right] \leq \delta
	\]
	where the first inequality follows from the definition of $\Sigma_{(i,b^{\bar{i}}, s^i , s^{\bar{i}})}$ and the second from Lemma \ref{lem:chernoff}. 
	
	All together we get $\mathrm{Pr}_{\vec{a},\vec{q}\sim\mathrm{O}^{\otimes n}_{AQ}} \left[ \mathcal{T}_{(i,b^{\bar{i}}, s^i , s^{\bar{i}})} (\vec{q^t},\vec{a^t})=1 \land \mathrm{O}^{\mathrm{EST2}}_{A|Q} \notin \Sigma_{(i,b^{\bar{i}}, s^i , s^{\bar{i}})}  \right] \leq 2\delta$ as required for the first part of the lemma.  
	
	We now proceed to the second part of the lemma. 
	
	If $\mathrm{Pr}_{\vec{a},\vec{q}\sim\mathrm{O}^{\otimes n}_{AQ}} \left[ \mathrm{O}^{\mathrm{EST2}}_{A|Q} \in \sigma_{(i,b^{\bar{i}}, s^i , s^{\bar{i}})} \right] \leq \delta$ then we are done. 
	Consider therefore states $\mathrm{O}_{A|Q}$ such that 
	\[
		\mathrm{Pr}_{\vec{a},\vec{q}\sim\mathrm{O}^{\otimes n}_{AQ}} \left[ \mathrm{O}^{\mathrm{EST2}}_{A|Q} \in \sigma_{(i,b^{\bar{i}}, s^i , s^{\bar{i}})} \right] > \delta\;.
	\]
	Using Lemma \ref{lem:chernoff} we know that there exists a state $\mathrm{O}^{\mathrm{EST2}}_{A|Q}\in \sigma_{(i,b^{\bar{i}}, s^i , s^{\bar{i}})}$ such that $|\mathrm{O}^{\mathrm{EST2}}_{A|Q}-\mathrm{O}_{A|Q}|_1 \leq \epsilon$ and according to the definition of $\sigma_{(i,b^{\bar{i}}, s^i , s^{\bar{i}})}$ this implies that $\mathrm{O}_{A|Q}$ is $\zeta$ signalling or more. Therefore, according to Lemma \ref{lem:reliable_test}, $\mathrm{Pr}_{\vec{a},\vec{q}\sim\mathrm{O}^{\otimes n}_{AQ}} \left[ \mathcal{T}_{(i,b^{\bar{i}}, s^i , s^{\bar{i}})} (\vec{q^t},\vec{a^t})=0  \right] \leq \delta$. All together we get  
	\[
		\mathrm{Pr}_{\vec{a},\vec{q}\sim\mathrm{O}^{\otimes n}_{AQ}} \left[ \mathcal{T}_{(i,b^{\bar{i}}, s^i , s^{\bar{i}})} (\vec{q^t},\vec{a^t})=0 \land \mathrm{O}^{\mathrm{EST2}}_{A|Q} \in \sigma_{(i,b^{\bar{i}}, s^i , s^{\bar{i}})}  \right] \leq 2\delta \;. \qedhere
	\]
	
\end{proof}

\end{document}